\documentclass[12pt]{article}
\usepackage{graphicx}
\usepackage{amsmath}
\usepackage{amsfonts}
\usepackage{amssymb}
\usepackage{appendix}
\usepackage{color}
\usepackage{enumitem}
\RequirePackage[colorlinks,linkcolor=red!50!black,citecolor=blue!50!black,urlcolor=blue!50!black]{hyperref}

\usepackage{tikz}
\usetikzlibrary{decorations.markings}
\usepackage{physics}
\usepackage{amsthm}
\usepackage{subcaption}
\usepackage{mathtools}
\usepackage{authblk}
\usepackage[margin=1.25in]{geometry}

\usepackage{xcolor} 
\usepackage{ifthen} 

\newcommand{\be}{\begin{equation}}
\newcommand{\ee}{\end{equation}}
\newcommand{\ba}{\begin{eqnarray}}
\newcommand{\ea}{\end{eqnarray}}

\newcommand{\term}[1]{\textbf{#1}}

\newcommand{\bs}[1]{\left[ #1 \right]} 
\newcommand{\br}[1]{\left( #1 \right)} 
\newcommand{\bc}[1]{\left\{ #1 \right\}} 

\newcommand{\wor}[2][]{{\mathtt{wor}_{#1}}\!\left(#2\right)}     
\newcommand{\obs}[2][]{{\mathtt{obs}_{#1}}\!\left(#2\right)}     
\newcommand{\pa}[2][]{{\mathtt{pa}_{#1}}\!\left(#2\right)}       
\newcommand{\vpa}[2][]{{\mathtt{vpa}_{#1}}\!\left(#2\right)}     
\newcommand{\lpa}[2][]{{\mathtt{lpa}_{#1}}\!\left(#2\right)}     
\newcommand{\ch}[2][]{{\mathtt{ch}_{#1}}\!\left(#2\right)}       
\newcommand{\aff}[2][]{{\mathtt{aff}_{#1}}\!\left(#2\right)}     
\newcommand{\an}[2][]{{\mathtt{an}_{#1}}\!\left(#2\right)}       
\newcommand{\des}[2][]{{\mathtt{des}_{#1}}\!\left(#2\right)}     
\newcommand{\sub}[2][]{{\mathtt{sub}_{#1}}\!\left(#2\right)}     
\newcommand{\exo}[2][]{{\mathtt{exo}_{#1}}\!\left(#2\right)}     
\newcommand{\perm}[2][]{{\mathtt{perm}_{#1}}\!\left(#2\right)}   

\newsavebox\MBox
\newcommand\colorline[2]{{\sbox\MBox{$#2$}%
  \rlap{\usebox\MBox}\color{#1}\rule[-1.4\dp\MBox]{\wd\MBox}{1.2pt}}}

\theoremstyle{plain}
\newtheorem{theorem}{Theorem}
\newtheorem{lemma}[theorem]{Lemma}

\theoremstyle{definition}
\newtheorem{definition}{Definition}
\newtheorem{example}{Example}

\theoremstyle{remark}

\usepackage[
    backend=biber,
    firstinits=true,
    doi=false,isbn=false,url=false,
    style=numeric,
    sorting=nyt,
    sortcites
]{biblatex}
\begin{document}

\title{A Combinatorial Solution to Causal Compatibility}
\author{Thomas C. Fraser \\ \href{mailto://tfraser@perimeterinstitute.ca}{tfraser@perimeterinstitute.ca}}
\affil{{\small Perimeter Institute for Theoretical Physics, Waterloo, Ontario, Canada, N2L 2Y5 \\ University of Waterloo, Waterloo, Ontario, Canada, N2L 3G1}}
\date{\today}

\maketitle
\begin{abstract}
    Within the field of causal inference, it is desirable to learn the structure of causal relationships holding between a system of variables from the correlations that these variables exhibit; a sub-problem of which is to certify whether or not a given causal hypothesis is \textit{compatible} with the observed correlations. A particularly challenging setting for assessing causal compatibility is in the presence of partial information; i.e. when some of the variables are hidden/latent. This paper introduces the possible worlds framework as a method for deciding causal compatibility in this difficult setting. We define a graphical object called a possible worlds diagram, which compactly depicts the set of all possible observations. From this construction, we demonstrate explicitly, using several examples, how to prove causal incompatibility. In fact, we use these constructions to prove causal incompatibility where no other techniques have been able to. Moreover, we prove that the possible worlds framework can be adapted to provide a \textit{complete} solution to the possibilistic causal compatibility problem. Even more, we also discuss how to exploit graphical symmetries and cross-world consistency constraints in order to implement a hierarchy of necessary compatibility tests that we prove converges to sufficiency.
\end{abstract}


\clearpage
\section{Introduction}
\label{sec:intro}

A theory of causation specifies the effects of actions with absolute necessity. On the other hand, a probabilistic theory encodes degrees of belief and makes predictions based on limited information. A common fallacy is to interpret correlation as causation; opening an umbrella has never caused it to rain, although the two are strongly correlated. Numerous paradoxical and catastrophic consequences are unavoidable when probabilistic theories and theories of causation are confused. Nonetheless, \textit{Reichenbach's principle} asserts that correlations must admit causal explanation; after all, the fear of getting wet causes one to open an umbrella.

In recent decades, a concerted effort has been put into developing a formal theory for probabilistic causation~\cite{Pearl_2009,Spirtes_2000}. Integral to this formalism is the concept of a \textit{causal structure}. A causal structure is a directed acyclic graph, or DAG, which encodes hypotheses about the causal relationships among a set of random variables. A \textit{causal model} is a causal structure when equipped with an explicit description of the parameters which govern the causal relationships. Given a multivariate probability distribution for a set of variables and a proposed causal structure, the \textit{causal compatibility problem} aims to determine the existence or non-existence of a causal model for the given causal structure which can explain the correlations exhibited by the variables. More generally, the objective of \textit{causal discovery} is to enumerate \textit{all} causal structure(s) compatible with an observed distribution. Perhaps unsurprisingly, causal inference has applications in a variety of academic disciplines including economics, risk analysis, epidemiology, bioinformatics, and machine learning~\cite{Pearl_2009,Pearl_2009_tr,Goudet_2017,Jing_2004,Robins_2000}.

For physicists, a consideration of causal influence is commonplace; the theory of special/general relativity strictly prohibits causal influences between space-like separated regions of space-time~\cite{Wald_2010}. Famously, in response to Einstein, Podolsky, and Rosen's~\cite{EPR_Orig} critique on the completeness of quantum theory, Bell~\cite{Bell_1964} derived an observational constraint, known as Bell's inequality, which must be satisfied by all hidden variable models which respect the causal hypothesis of relativity. Moreover, Bell demonstrated the existence of quantum-realizable correlations which violate Bell's inequality~\cite{Bell_1964}. Recently, it has been appreciated that Bell's theorem can be understood as an instance of causal inference~\cite{Wood_2012}. Contemporary quantum foundations maintains two closely related causal inference research programs. The first is to develop a theory of \textit{quantum} causal models in order to facilitate a causal description of quantum theory and to better understand the limitations of quantum resources~\cite{Henson_2014,Fritz_2014,Leifer_2011,Ried_2015,Allen_2016,Chaves_2015,Costa_2016,Barret_2007,Wolfe_2016,Pienaar_2015,Oreshkov_2011}. The second is the continued study of \textit{classical} causal inference with the purpose of distinguishing genuinely quantum behaviors from those which admit classical explanations~\cite{Wolfe_2016,Fritz_2012,Fritz_2014,Abramsky_2012,Weilenmann_2016,Rosset_2017,Abramsky_2011,Chaves_2016,Fraser_2017}. In particular, the results of \cite{Henson_2014} suggest that causal structures which support quantum non-classicality are uncommon and typically large in size; therefore, systematically finding new examples of such causal structures will require the development of new algorithmic strategies. As a consequence, quantum foundations research has relied upon, and contributed to, the techniques and tools used within the field of causal inference~\cite{Wolfe_2016,Rosset_2017,Chaves_2015,Henson_2014}. The results of this paper are concerned exclusively with the latter research program of \textit{classical} causal inference, but does not rule out the possibility of a generalization to quantum causal inference.

When all variables in a probabilistic system are observed, checking the compatibility status between a joint distribution and a causal structure is relatively easy; compatibility holds if and only if all conditional independence constraints implied by graphical \textit{d-separation} relations hold~\cite{Pearl_2009,Pearl_2013}. Unfortunately, in more realistic situations there are ethical, economic, or fundamental barriers preventing access to certain statistically relevant variables, and it becomes necessary to hypothesize the existence of \textit{latent/hidden} variables in order to adequately explain the correlations expressed by the \textit{visible/observed} variables~\cite{Pearl_2009,Evans_2015,Wolfe_2016}.
In the presence of latent variables, and in the absence of interventional data, the causal compatibility problem, and by extension the subject of causal inference as a whole, becomes considerably more difficult.

In order to overcome these difficulties, numerous simplifications have be invoked by various authors in order to make partial progress. A particularly popular simplification strategy has been to consider alternative classes of graphical causal models which can act as surrogates for DAG causal models; e.g. MC-graphs~{\cite{Koster_2002}}, summary graphs~{\cite{Wermuth_2011}}, or maximal ancestral graphs (MAGs)~{\cite{Richardson_2002,Zhang_2008}}. While these approaches are certainly attractive from a practical perspective (efficient algorithms such as FCI~{\cite{Spirtes_2000}} or RFCI~{\cite{Colombo_2012}} exist for assessing causal compatibility with MAGs, for instance), they nevertheless fail to fully capture all constraints implied by DAG causal models with latent variables~{\cite{Evans_2016}}.\footnote{For concrete and relevant example of this weakness, note that there are observable distributions incompatible with the DAG causal structure in Figure~{\ref{fig:triangle_structure}} (which admits of no observable d-separation relations), whereas its associated MAG is compatible with \textit{all} observed distributions. An analogous statement happens to be true of the DAG causal structure in Figure~{\ref{fig:evans_causal_structure_1}}.} The forthcoming formalism is concerned with assessing the causal compatibility of DAG causal structures \textit{directly}, therefore avoiding these shortcomings.

Nevertheless, when considering DAG causal structures directly (henceforth just causal structures), making assumptions about the nature of the latent variables and the parameters which govern them can simplify the problem~\cite{Spirtes_2013,Wainwright_2007,Geiger_2013}. For instance, when the latent variables are assumed to have a known and finite cardinality\footnote{The cardinality of a random variable is the size of its sample space.}, it becomes possible to articulate the causal compatibility problem as a finite system of polynomial equality and inequality equations with a finite list of unknowns for which non-linear quantifier elimination methods, such as Cylindrical Algebraic Decomposition~\cite{Jirstrand_1995}, can provide a complete solution. Unfortunately, these techniques are only computationally tractable in the simplest of situations. Other techniques from algebraic geometry have been used in simple scenarios to approach the causal compatibility problem as well~\cite{Lee_2015,Garcia_2003,Geiger_2013}. When no assumptions about the nature of the latent variables are made, there are a plethora of methods for deriving novel equality~\cite{Richardson_2017,Evans_2015} and inequality~\cite{Wolfe_2016,Fritz_2011,Abramsky_2012,Chaves_2016,Weilenmann_2016,Evans_2012,Bancal_2010,Henson_2014,Fraser_2017,Bonet_2013,Studel_2015} constraints that must be satisfied by any compatible distribution. The majority of these methods are unsatisfactory on the basis that the derived constraints are necessary, but not sufficient. A notable exception is the Inflation Technique~\cite{Wolfe_2016}, which produces a hierarchy of linear programs (solvable using efficient algorithms~\cite{Bradley_1977,Jones_2004,Schrijver_1998,Dantzig_1973,Kavvadias_2005}) which are necessary \textit{and} sufficient~\cite{Navascues_2017} for determining compatibility.

In contrast with the aforementioned algebraic techniques, the purpose of this paper is to present the \textit{possible worlds framework}, which offers a combinatorial solution to the causal compatibility problem in the presence of latent variables. Importantly, this framework can only be applied when the cardinality of the visible variables are known to be finite.\footnote{Regarding the latent variables, Appendix~{\ref{sec:finite_bounds}} demonstrates that the latent variables can be assumed to have finite cardinality without loss of generality whenever the visible variables have finite cardinality.} This framework is inspired by the twin networks of Pearl~\cite{Pearl_2009}, parallel worlds of Shpitser~\cite{Shpitser_2008}, and by some original drafts of the Inflation Technique paper~\cite{Wolfe_2016}. The possible worlds framework accomplishes three things. First, we prove its conceptual advantages by revealing that a number of disparate instances of causal incompatibility become unified under the same premise. Second, we provide a closed-form algorithm for completely solving the \textit{possibilistic} causal compatibility problem. To demonstrate the utility of this method, we provide a solution to an unsolved problem originally reported~\cite{Evans_2016}. Third, we show that the possible worlds framework provides a hierarchy of tests, much like the Inflation Technique, which solves completely the \textit{probabilistic} causal compatibility problem. 

Unfortunately, the computational complexity of the proposed probabilistic solution is prohibitively large in many practical situations. Therefore, the contributions of this work are primarily conceptual. Nevertheless, it is possible that these complexity issues are intrinsic to the problem being considered.
Notably, the hierarchy of tests presented here has an asymptotic rate of convergence commensurate to the only other complete solution to the probabilistic compatibility problem, namely the hierarchy of tests provided in~\cite{Navascues_2017}. Moreover, unlike the Inflation Technique, if a distribution is compatible with a causal structure, then the hierarchy of tests provided here has the advantage of returning a causal model which generates that distribution.

This paper is organized as follows: Section~\ref{sec:a_review_of_causal_modeling} begins with a review of the mathematical formalism behind causal modeling, including a formal definition of the causal compatibility problem, and also introduces the notations to be used throughout the paper. Afterwards, Section~\ref{sec:possible_worlds_framework} introduces the possible worlds framework and defines its central object of study: a possible worlds diagram. Section~\ref{sec:complete_possibilistic} applies the possible worlds framework to prove possibilistic incompatibility between several distributions and corresponding causal structures, culminating in an algorithm for exactly solving the possibilistic causal compatibility problem. Finally, Section~\ref{sec:complete_probabilistic} establishes a hierarchy of tests which completely solve the probabilistic causal compatibility problem. Moreover, Section~{\ref{sec:symmetry_and_superfluous}} articulates how to utilize internal symmetries in order to alleviate the aforementioned computational complexity issues. Section~\ref{sec:conclusion} concludes.

Appendix~\ref{sec:simplifying_causal_structures} summarizes relevant results from~\cite{Evans_2016} needed in Section~\ref{sec:a_review_of_causal_modeling}. Appendix~\ref{sec:simplifying_causal_parameters} generalizes the results of~\cite{Rosset_2017}, placing new upper bounds on the maximum cardinality of the latent variables, required for Sections~\ref{sec:a_review_of_causal_modeling} and~\ref{sec:complete_probabilistic}.

\section{A Review of Causal Modeling}
\label{sec:a_review_of_causal_modeling}

This review section is segmented into three portions. First, Section~\ref{sec:directed_graphs} defines directed graphs and their properties. Second, Section~\ref{sec:probability_distributions} introduces the notation and terminology regarding probability distributions to be used throughout the remainder of this article. Finally, Section~\ref{sec:causal_models} defines the notion of a causal model and formally introduces the causal compatibility problem.

\subsection{Directed Graphs}
\label{sec:directed_graphs}

\begin{definition}
    A \term{directed graph} $\mathcal G$ is an ordered pair $\mathcal G = \br{\mathcal Q, \mathcal E}$ where $\mathcal Q$ is a finite set of \textit{vertices} and $\mathcal E$ is a set \textit{edges}, i.e. ordered pairs of vertices $\mathcal E \subseteq \mathcal Q \times \mathcal Q$. If $\br{q,u} \in \mathcal E$ is an edge, denoted as $q \to u$, then $u$ is a \term{child} of $q$ and $q$ is a \term{parent} of $u$. A \term{directed path} of length $k$ is a sequence of vertices $q_{(1)} \to q_{(2)} \to \cdots \to q_{(k)}$ connected by directed edges. For a given vertex $q$, $\pa[\mathcal G]{q}$ denotes its parents and $\ch[\mathcal G]{q}$ its children. If there is a directed path from $q$ to $u$ then $q$ is an \term{ancestor} of $u$ and $u$ is a \term{descendant} of $q$; the set of all ancestors of $q$ is denoted $\an[\mathcal G]{q}$ and the set of all descendants is denoted $\des[\mathcal G]{q}$. The definition for parents, children, ancestors and descendants of a single vertex $q$ are applied disjunctively to sets of vertices $Q \subseteq \mathcal Q$:
    \begin{align}
        \ch[\mathcal G]{Q} &= \bigcup_{q \in Q} \ch[\mathcal G]{q}, \qquad \pa[\mathcal G]{Q} = \bigcup_{q \in Q} \pa[\mathcal G]{q}, \\
        \an[\mathcal G]{Q} &= \bigcup_{q \in Q} \an[\mathcal G]{q}, \qquad \des[\mathcal G]{Q} = \bigcup_{q \in Q} \des[\mathcal G]{q}.
    \end{align}
    A directed graph is \term{acyclic} if there is no directed path of length $k > 1$ from $q$ back to $q$ for any $q \in \mathcal{Q}$ and \term{cyclic} otherwise. For example, Figure~\ref{fig:examples_of_directed_graphs} depicts the difference between cyclic and acyclic directed graphs.
\end{definition}

\begin{definition}
    The \term{subgraph} of $\mathcal G = \br{\mathcal Q, \mathcal E}$ induced by $\mathcal W \subset \mathcal Q$, denoted $\sub[\mathcal G]{\mathcal W}$, is given by,
    \begin{align}
        \sub[\mathcal G]{\mathcal W} = \br{\mathcal W, \mathcal E \cap \br{\mathcal W \times \mathcal W}},
    \end{align}
    i.e. the graph obtained by taking all edges from $\mathcal E$ which connect members of $\mathcal W$.
\end{definition}

\begin{figure}
    \centering
    \begin{subfigure}[t]{0.49\textwidth}
        \centering
        \includegraphics{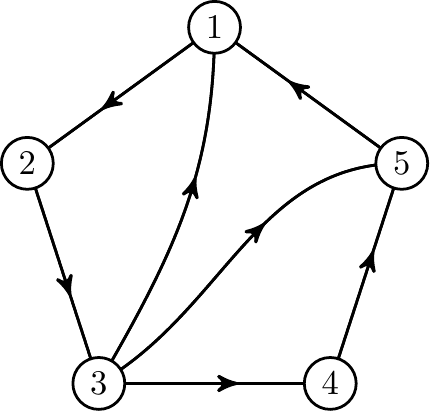}
        \caption{A directed cyclic graph.}
    \end{subfigure}
    \begin{subfigure}[t]{0.49\textwidth}
        \centering
        \includegraphics{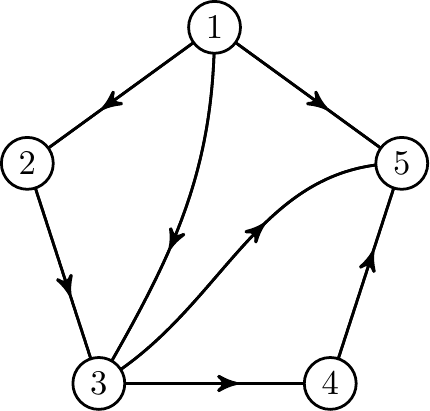}
        \caption{A directed acyclic graph.}
    \end{subfigure}
    \caption{The difference between a directed \textit{cyclic} graph and a directed \textit{acyclic} graph.}
    \label{fig:examples_of_directed_graphs}
\end{figure}

\subsection{Probability Theory}
\label{sec:probability_distributions}

\begin{definition}[Probability Theory]
    A \term{probability space} is a triple $\br{\Omega, \Xi, \mathtt{P}}$ where the state space $\Omega$ is the set of all possible \textit{outcomes}, $\Xi \subseteq 2^{\Omega}$ is the set of \textit{events} forming a $\sigma$-algebra over $\Omega$, and $\mathtt{P}$ is a $\sigma$-additive function from events to probabilities such that $\mathtt{P}(\Omega) = 1$.
\end{definition}

\begin{definition}[Probability Notation]
    For a collection of random variables $\mathsf{X}_{\mathcal I} = \bc{\mathsf{X}_{1},\mathsf{X}_{2},\ldots, \mathsf{X}_{k}}$ indexed by $i \in \mathcal I = \bc{1,2,\ldots, k}$ where each $\mathsf{X}_{i}$ takes values from $\Omega_{i}$, a joint distribution $\mathtt{P}_{\mathcal I} = \mathtt{P}_{12\ldots k}$ assigns probabilities to outcomes from $\Omega_{\mathcal I} = \prod_{i \in \mathcal I} \Omega_{i}$. The event that each $\mathsf{X}_{i}$ takes value $x_{i}$, referred to as a \term{valuation} of $\mathsf X_{\mathcal I}$\footnote{A valuation is a particular type of event in $\Xi$ where the random variables take on definite values.}, is denoted as,
    \begin{align}
        \mathtt{P}_{\mathcal I}(x_{\mathcal I}) = \mathtt{P}_{12\ldots k}\br{x_1x_2\ldots x_k} = \mathtt{P}\br{X_{1} = x_1, X_{2} = x_2, \ldots X_{k} = x_k}.
        \label{eq:valuation_definition}
    \end{align}
    A point distribution $\mathtt{P}_{\mathcal I}(y_{\mathcal I}) = 1$ for a particular event $y_{\mathcal I} \in \Omega_{\mathcal I}$ is expressed using square brackets,
    \begin{align}
        \mathtt{P}_{\mathcal I}(y_{\mathcal I}) = 1 \Leftrightarrow \mathtt{P}_{\mathcal I}(x_{\mathcal I}) = [y_{\mathcal I}](x_{\mathcal I}) = \delta(y_{\mathcal I}, x_{\mathcal I}) = \prod_{i \in \mathcal I} \delta(y_{i}, x_{i}).
    \end{align}
    The set of all probability distributions over $\Omega_{\mathcal I}$ is denoted as $\mathbb{P}_{\mathcal I}$. Let $k_{i}$ denote the \textit{cardinality} or size of $\Omega_{i}$. If $\mathsf{X}_{i}$ is discrete, then $k_{i} = \abs{\Omega_{i}}$, otherwise $\mathsf{X}_{i}$ is continuous and $k_{i} = \infty$.
\end{definition}

\subsection{Causal Models and Causal Compatibility}
\label{sec:causal_models}

A \term{causal model} represents a complete description of the causal mechanisms underlying a probabilistic process. Formally, a causal model is a pair of objects $\br{\mathcal G, \mathcal P}$, which will be defined in turn. First, $\mathcal G$ is a directed acyclic graph $\br{\mathcal Q, \mathcal E}$, whose vertices $q \in \mathcal Q$ represent random variables $\mathsf{X}_{\mathcal Q} = \{ \mathsf{X}_{q} \mid q \in \mathcal Q \}$. The purpose of a causal structure is to graphically encode the causal relationships between the variables. Explicitly, if $q \to u \in \mathcal E$ is an edge of the causal structure, $\mathsf{X}_{q}$ is said to have \textit{causal influence} on $\mathsf{X}_{u}$\footnote{It is seldom necessary to make the distinction between the random variable $\mathsf{X}_{q}$ and the index/vertex $q$; this paper henceforth treats them as synonymous.}. Consequently, the causal structure predicts that given complete knowledge of a valuation of the parental variables $\mathsf{X}_{\pa[\mathcal G]{u}} = \bc{\mathsf{X}_{q} \mid q \in \pa[\mathcal G]{u}}$, the random variable $\mathsf{X}_{u}$ should become independent of its non-descendants\footnote{This is known as the local Markov property.}~\cite{Pearl_2009}. With this observation as motivation, the \term{causal parameters} $\mathcal P$ of a causal model are a family of conditional probability distributions $\mathtt{P}_{q |\pa[\mathcal G]{q}}$ for each $q \in \mathcal Q$. In the case that $q$ has no parents in $\mathcal G$, the distribution is simply unconditioned. The purpose of the causal parameters are to predict a joint distribution $\mathtt{P}_{\mathcal Q}$ on the configurations $\Omega_{\mathcal Q}$ of a causal structure,
\begin{align}
    \forall x_{\mathcal Q} \in \Omega_{\mathcal Q}, \quad \mathtt{P}_{\mathcal Q}(x_{\mathcal Q}) = \prod_{q\in \mathcal Q} \mathtt{P}_{q | \pa[\mathcal G]{q}}\!(x_{q} | x_{\pa[\mathcal G]{q}}). \label{eq:causal_parameter_product}
\end{align}
If the hypotheses encoded within a causal structure $\mathcal G$ are correct, then the observed distribution over $\Omega_{\mathcal Q}$ should factorize according to Equation~\ref{eq:causal_parameter_product}. Unfortunately, as discussed in Section~\ref{sec:intro}, there are often ethical, economic, or fundamental obstacles preventing access to all variables of a system. In such cases, it is customary to partition the vertices of causal structure into two disjoint sets; the \term{visible (observed) vertices} $\mathcal V$, and the \term{latent (unobserved) vertices} $\mathcal L$ (for example, see Figure~\ref{fig:example_causal_structure_0}). Additionally, we denote visible parents of any vertex $q \in \mathcal V \cup \mathcal L$ as $\vpa[\mathcal G]{q} = \mathcal V \cap \pa[\mathcal G]{q}$ and analogously for the latent parents $\lpa[\mathcal G]{q} = \mathcal L \cap \pa[\mathcal G]{q}$.

\begin{figure}
    \centering
    \includegraphics{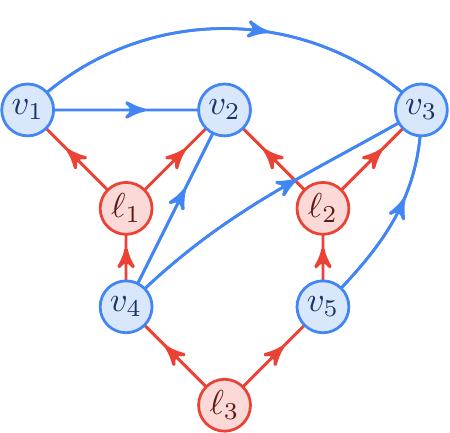}
    \caption{The causal {structure} $\mathcal G_{\ref{fig:example_causal_structure_0}}$ in this figure encodes a causal hypothesis about the causal relationships between the visible variables $\mathcal V = \bc{v_{1}, v_{2}, v_{3}, v_{4}, v_{5}}$ and the latent variables $\mathcal L = \bc{\ell_{1}, \ell_{2}, \ell_{3}}$; e.g. $v_{2}$ experiences a direct causal influence from each of its parents, both visible $\vpa[\mathcal G_{\ref{fig:example_causal_structure_0}}]{v_{2}} = \bc{v_{1}, v_{4}}$ and latent $\lpa[\mathcal G_{\ref{fig:example_causal_structure_0}}]{v_{2}} = \bc{\ell_{1}, \ell_{2}}$. Throughout this {paper}, visible variables and edges connecting them are colored blue whereas all latent variables and all other edges are colored red. }
    \label{fig:example_causal_structure_0}
\end{figure}

In the presence of latent variables, Equation~\ref{eq:causal_parameter_product} stills makes a prediction about the joint distribution $\mathtt{P}_{\mathcal V \cup \mathcal L}(x_{\mathcal V}, \lambda_{\mathcal L})$\footnote{This paper adopts the notational convenient of using $\lambda_{\ell} \in \Omega_{\ell}$ for valuations of latent variables $\ell \in \mathcal L$ to differentiate them from valuations $x_{v} \in \Omega_{v}$ of observed variables $v \in \mathcal V$.} over the visible and latent variables, albeit an experimenter attempting to verify or discredit a causal hypothesis only has access to the marginal distribution $\mathtt{P}_{\mathcal V}(x_{\mathcal V})$. If $\Omega_{\mathcal L}$ is continuous,
\begin{align}
    \forall x_{\mathcal V} \in \Omega_{\mathcal V}, \mathtt{P}_{\mathcal V}(x_{\mathcal V}) = \int_{\lambda_{\mathcal L} \in \Omega_{\mathcal L}} \mathrm{d}\mathtt{P}_{\mathcal V \cup \mathcal L}(x_{\mathcal V}, \lambda_{\mathcal L})
    \label{eq:latent_marginalization_continuous}
\end{align}
If $\Omega_{\mathcal L}$ is discrete,
\begin{align}
    \forall x_{\mathcal V} \in \Omega_{\mathcal V}, \mathtt{P}_{\mathcal V}(x_{\mathcal V}) = \sum_{\lambda_{\mathcal L} \in \Omega_{\mathcal L}} \mathtt{P}_{\mathcal V \cup \mathcal L}(x_{\mathcal V}, \lambda_{\mathcal L}).
    \label{eq:latent_marginalization_discrete}
\end{align}
A natural question arises; in the absence of information about the latent variables $\mathcal L$, how can one determine whether or not their causal hypotheses are correct? The principle purpose of this paper is to provide the reader with methods for answering this question.

In general, other than being a directed acyclic graph, there are no restrictions placed on a causal structure with latent variables. Nonetheless,~\cite{Evans_2016} demonstrates that every causal structure $\mathcal G$ can be converted into a standard form that is observationally equivalent to $\mathcal G$ where the latent variables are exogenous (have no parents) and whose children sets are isomorphic to the facets of a simplicial complex over $\mathcal V$\footnote{Appendix~\ref{sec:observational_impact} briefly discusses what it means for two causal structures to be \textit{observationally equivalent}.}. Appendix~\ref{sec:simplifying_causal_structures} summarizes the relevant results from~\cite{Evans_2016} necessary for making this claim. Additionally, Appendix~\ref{sec:simplifying_causal_parameters} demonstrates that any finite distribution $\mathtt{P}_{\mathcal V}$ which satisfies the causal hypotheses (i.e. Equation~\ref{eq:latent_marginalization_continuous}) can be generated using deterministic causal parameters for the visible variables and moreover, the cardinalities of the latent variables can be assumed finite\footnote{We prove this result in Appendix~\ref{sec:simplifying_causal_parameters} by generalizing the proof techniques used in~\cite{Rosset_2017}.}. Altogether, Appendices~\ref{sec:simplifying_causal_structures} and~\ref{sec:simplifying_causal_parameters} suggest that without loss of generality, we can simplify the causal compatibility problem as follows:

\begin{definition}[Functional Causal Model]
    \label{defn:functional_causal_model}
    A \term{(finite) functional causal model} for a causal structure $\mathcal G = \br{\mathcal V \cup \mathcal L, \mathcal E}$ is a triple $\br{\mathcal G, \mathcal F_{\mathcal V}, \mathcal P_{\mathcal L}}$ where
    \begin{align}
        \mathcal F_{\mathcal V} = \{f_{v} : \Omega_{\pa[\mathcal G]{v}} \to \Omega_{v} \mid v \in \mathcal V\}
        \label{eq:functional_parameters}
    \end{align}
    are deterministic functions for the visible variables $\mathcal V$ in $\mathcal G$, and
    \begin{align}
        \mathcal P_{\mathcal L} = \bc{\mathtt{P}_{\ell} : \Omega_{\ell} \to \bs{0,1} \mid \ell \in \mathcal L}
        \label{eq:latent_parameters}
    \end{align}
    are finite probability distributions for the latent variables $\mathcal L$ in $\mathcal G$. A functional causal model defines a probability distribution $\mathtt{P}_{\mathcal V} : \Omega_{\mathcal V} \to \bs{0,1}$,
    \begin{align}
        \forall x_{\mathcal V} \in \Omega_{\mathcal V}, \quad \mathtt{P}_{\mathcal V}(x_{\mathcal V}) = \prod_{\ell \in \mathcal L} \sum_{\lambda_{\ell} \in \Omega_{\ell}} \mathtt{P}_{\ell}(\lambda_{\ell}) \prod_{v \in \mathcal L} \delta(x_{v}, f_{v}(x_{\vpa[\mathcal G]{v}},\lambda_{\lpa[\mathcal G]{v}})). \label{eq:functional_compatibility}
    \end{align}
\end{definition}

\begin{definition}[The Causal Compatibility Problem]
    \label{defn:causal_compat}
    Given a causal structure $\mathcal G = \br{\mathcal V \cup \mathcal L, \mathcal E}$ and a distribution $\mathtt{P}_{\mathcal V}$ over the visible variables $\mathcal V$, \term{the causal compatibility problem} is to determine if there exists a functional causal model $\br{\mathcal G, \mathcal F_{\mathcal V}, \mathcal P_{\mathcal L}}$ (defined in Definition~\ref{defn:functional_causal_model}) such that Equation~\ref{eq:functional_compatibility} reproduces $\mathtt{P}_{\mathcal V}$. If such a functional causal model exists, then $\mathtt{P}_{\mathcal V}$ is said to be \term{compatible} with $\mathcal G$; otherwise $\mathtt{P}_{\mathcal V}$ is incompatible with $\mathcal G$. The set of all \textit{compatible} distributions on $\mathcal V$ for a causal structure $\mathcal G$ is denoted $\mathcal M_{\mathcal V}\br{\mathcal G}$.
\end{definition}

\section{The Possible Worlds Framework}
\label{sec:possible_worlds_framework}

Consider the causal structure in Figure~\ref{fig:example_causal_structure} denoted $\mathcal G_{\ref{fig:example_causal_structure}}$. For the sake of concreteness, suppose one is promised the latent variables are sampled from a binary sample space, i.e. $k_{\mu} = k_{\nu} = 2$. Let $z_{\mu} = \mathtt{P}_{\mu}(0_{\mu})$ and $z_{\nu} = \mathtt{P}_{\nu}(0_{\nu})$. The causal hypothesis $\mathcal G_{\ref{fig:example_causal_structure}}$ predicts (via Equation~\ref{eq:functional_compatibility}) that observable events $(x_{a}, x_{b}, x_{c}) \in \Omega_{a} \times \Omega_{b} \times \Omega_{c}$ will be distributed according to,
\begin{align}
\begin{split}
    \mathtt{P}_{abc} &= z_{\mu}z_{\nu}[\obs[abc]{0_{\mu}0_{\nu}}] + z_{\mu}(1-z_{\nu})[\obs[abc]{0_{\mu}1_{\nu}}] + \\
    &\quad+ (1-z_{\mu})z_{\nu}[\obs[abc]{1_{\mu}0_{\nu}}] + (1-z_{\mu})(1-z_{\nu})[\obs[abc]{1_{\mu}1_{\nu}}],
\end{split}
\end{align}
where $\obs[abc]{\lambda_{\mu}\lambda_{\nu}} \in \Omega_{a} \times \Omega_{b} \times \Omega_{c}$ is shorthand for the observed event generated by the autonomous functions $f_{a}, f_{b}, f_{c}$ for each $(\lambda_{\mu}, \lambda_{\nu}) \in \Omega_{\mu} \times \Omega_{\nu}$. In the case of $\mathcal G_{\ref{fig:example_causal_structure}}$,
\begin{align}
    \obs[abc]{\lambda_{\mu}\lambda_{\nu}} = (f_{a}(\lambda_{\mu}),f_{b}(f_{a}(\lambda_{\mu}),\lambda_{\nu}),f_{c}(f_{b}(f_{a}(\lambda_{\mu}),\lambda_{\nu}),\lambda_{\nu})).
\end{align}
For each distinct realization $(\lambda_{\mu}, \lambda_{\nu}) \in \Omega_{\mu} \times \Omega_{\nu}$ of the latent variables, one can consider a possible world wherein the values $\lambda_\mu, \lambda_\nu$ are \textit{not} sampled according to the respective distributions $\mathtt{P}_{\mu}, \mathtt{P}_{\nu}$, but instead take on definite values. From the perspective of counterfactual reasoning, each world is modelling a distinct counterfactual assignment of the latent variables, but \textit{not} the visible variables.\footnote{It is conceivable that this framework, and its associated diagrammatic notation, could be extended to accommodate counterfactual assignments to the visible variables as well. Such an extension could be useful for assessing compatibility with interventional data, in addition to the purely observational data being considered here.} In this particular example, there are $k_{\mu} \times k_{\nu} = 2 \times 2 = 4$ distinct, possible worlds. Figure~\ref{fig:example_causal_structure_definite_worlds} represents, and uniquely colors, these possible worlds. Note that the definite valuations of the latent variables in Figure~\ref{fig:example_causal_structure_definite_worlds} are depicted using squares\footnote{This diagrammatic convention is imminently explained in more depth by Definition~\ref{defn:awf} and associated Figure~\ref{fig:vertex_explained}.}. Critically, regardless of the deterministic functional relationships $f_{a}, f_{b}, f_{c}$, there are identifiable consistency constraints that must hold between these worlds. For example, $a$ is determined by a function $f_{a} : \Omega_{\mu} \to \Omega_{a}$ and thus the observed value for $a$ in the yellow $(0_{\mu}0_{\nu})$-world must be \textit{exactly} the same as the observed value for $a$ in the green $(0_{\mu}1_{\nu})$-world. This cross-world consistency constraint is illustrated in Figure~\ref{fig:example_causal_structure_merged} by embedding each possible world into a larger diagram with overlapping $\lambda_{\mu} \to a$ subgraphs. It is important to remark that not all cross-world consistency constraints are captured by this diagram; the value of $b$ in the yellow $(0_{\mu}0_{\nu})$-world \textit{must} match the value of $b$ in the orange $(1_{\mu}0_{\nu})$-world if the value of $a$ in both possible worlds is the same.

For comparison, in the original causal structure $\mathcal G_{\ref{fig:example_causal_structure}}$, the vertices represented random variables sampled from distributions associated with causal parameters; whereas in the possible worlds diagram of Figure~\ref{fig:example_causal_structure_merged}, every valuation, including the latent valuations are \textit{predetermined} by the functional dependences $f_{a}, f_{b}, f_{c}$. For example, Figure~\ref{fig:example_causal_structure_merged_valuation} populates Figure~\ref{fig:example_causal_structure_merged} with the observable events generated by the following functional dependences,
\begin{align}
\begin{split}
    f_{a}(0_{\mu}) = 0_{a} &\quad f_{a}(1_{\mu}) = 1_{a}, \\
    f_{b}(0_{a}0_{\nu}) = 3_{b} \quad f_{b}(0_{a}1_{\nu}) = 1_{b} &\quad f_{b}(1_{a}0_{\nu}) = 2_{b} \quad f_{b}(1_{a}1_{\nu}) = 0_{b}, \\
    f_{c}(3_{b}0_{\mu}0_{\nu}) = 0_{c} \quad f_{c}(1_{b}0_{\mu}1_{\nu}) = 1_{c} &\quad f_{c}(2_{b}1_{\mu}0_{\nu}) = 2_{c} \quad f_{c}(0_{b}1_{\mu}1_{\nu}) = 3_{c}.
\end{split}
\label{eq:deterministic_model_example}
\end{align}

\begin{figure}
    \centering
    \begin{subfigure}[t]{0.49\textwidth}
        \centering
        \includegraphics{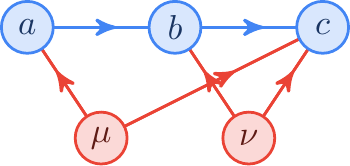}
        \caption{An example causal structure $\mathcal G_{\ref{fig:example_causal_structure}}$.}
        \label{fig:example_causal_structure}
    \end{subfigure}
    \begin{subfigure}[t]{0.49\textwidth}
        \centering
        \includegraphics{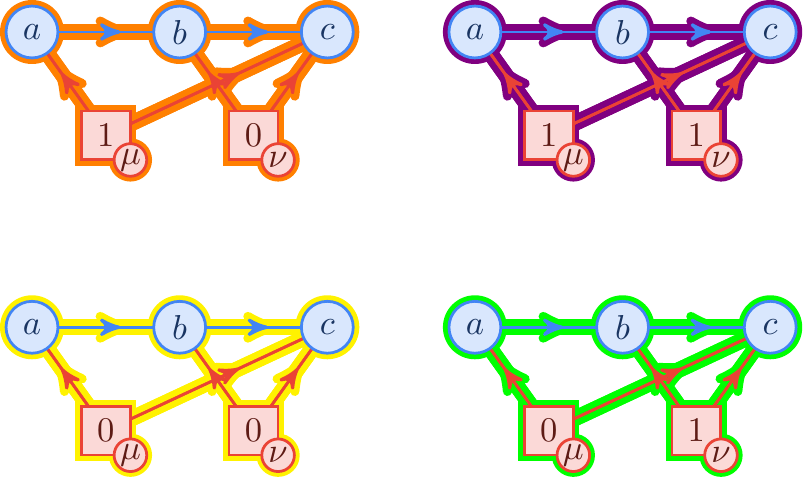}
        \caption{The possible worlds picture for $\mathcal G_{\ref{fig:example_causal_structure}}$.}
        \label{fig:example_causal_structure_definite_worlds}
    \end{subfigure}
    \begin{subfigure}[t]{0.49\textwidth}
        \centering
        \includegraphics{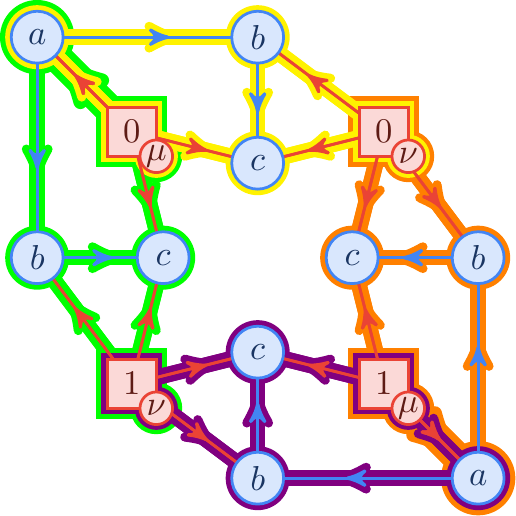}
        \caption{Identifying consistency constraints among possible worlds for $\mathcal G_{\ref{fig:example_causal_structure}}$.}
        \label{fig:example_causal_structure_merged}
    \end{subfigure}
    \begin{subfigure}[t]{0.49\textwidth}
        \centering
        \includegraphics{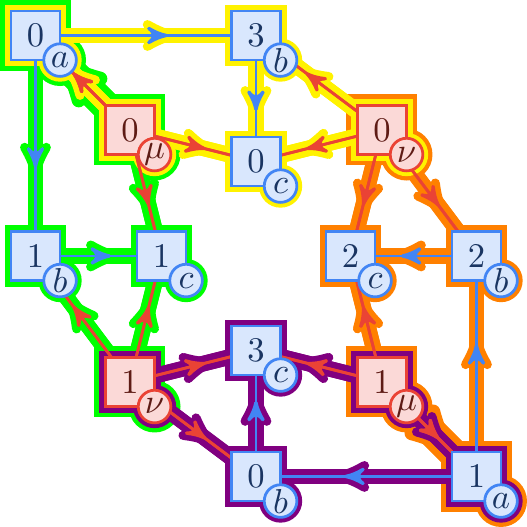}
        \caption{Populating a possible worlds diagram with the deterministic functions $f_{a}, f_{b}, f_{c}$ in Equation~\ref{eq:deterministic_model_example}.}
        \label{fig:example_causal_structure_merged_valuation}
    \end{subfigure}
    \caption{A causal structure $\mathcal G_{\ref{fig:example_causal_structure}}$ and the creation of the possible worlds diagram when $k_{\mu} = k_{\nu} = 2$.}
    \label{fig:d_structure_def}
\end{figure}

The utility of Figure~\ref{fig:example_causal_structure_merged_valuation} is in its simultaneous accounts of Equation~\ref{eq:deterministic_model_example}, the causal structure $\mathcal G_{\ref{fig:example_causal_structure}}$ and the cross-world consistency constraints that $\mathcal G_{\ref{fig:example_causal_structure}}$ induces. Nonetheless, Figure~\ref{fig:example_causal_structure_merged_valuation} fails to specify the probabilities $z_{\mu}, z_{\nu}$ associated with the latent events. In Section~\ref{sec:complete_possibilistic}, we utilize diagrams analogous to Figure~\ref{fig:example_causal_structure_merged_valuation} to tackle the causal compatibility problem. Before doing so, this paper needs to formally define the \textit{possible worlds framework}.

\begin{figure}
    \centering
    \includegraphics{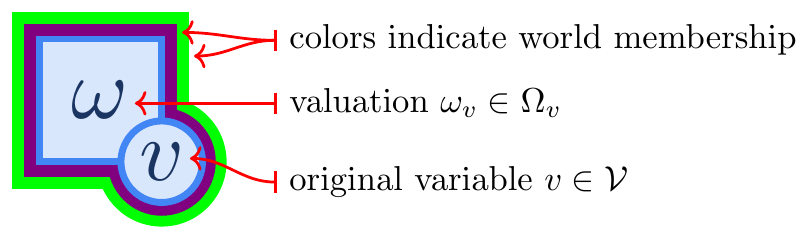}
    \caption{A vertex of a possible worlds diagram dissected.}
    \label{fig:vertex_explained}
\end{figure}
\begin{definition}[The Possible Worlds Framework]
    \label{defn:awf}
    Let $\mathcal G = \br{\mathcal V \cup \mathcal L, \mathcal E}$, be a causal structure with visible variables $\mathcal V$ and latent variables $\mathcal L$. Let $\mathcal F_{\mathcal V}$ be a set of functional parameters for $\mathcal V$ defined exactly as in Equation~\ref{eq:functional_parameters}. The \term{possible worlds diagram} for the pair $(\mathcal G, \mathcal F_{\mathcal V})$ is a directed acyclic graph $\mathcal D$ satisfying the following properties:
    \begin{enumerate}
        \item (Valuation Vertices) Each vertex in $\mathcal D$ consists of three pieces (consult Figure~\ref{fig:vertex_explained} for clarity):
        \begin{enumerate}
            \item a subscript $q \in \mathcal V \cup \mathcal L$ corresponding to a vertex in $\mathcal G$ (indicated inside a small circle in the bottom-right corner),
            \item an integer $\omega$ corresponding to a possible valuation/outcome $\omega_{q}$ of $q$ where $\omega_{q} \in \bc{0_{q}, 1_{q}, \ldots} = \Omega_{q}$ (indicated inside the square of each vertex),
            \item and a decoration in the form of colored outlines\footnote{The order of the colored outlines are arbitrary.} indicating which worlds (defined below) the vertex is a member of\footnote{Every valuation vertex belongs to at least one world.}.
        \end{enumerate}
        \item (Ancestral Isomorphism)\footnote{Readers who are familiar with the Inflation technique~{\cite{Wolfe_2016}} will recognize this ancestral isomorphism property from the definition of an \textit{Inflation} of a causal structure. The critical difference between a possible worlds diagram and an Inflation is that vertices in the former represent valuations of variables whereas vertices in the latter represent independent copies of the variables.} For every valuation vertex $\omega_{q}$ in $\mathcal D$, the ancestral subgraph of $\omega_{q}$ in $\mathcal D$ is isomorphic to the ancestral subgraph of $q$ in $\mathcal G$ under the map $\omega_{q} \mapsto q$.
        \begin{align}
            \label{eq:ancestral_matching}
            \sub[\mathcal D]{\an[\mathcal D]{\omega_{q}}} \simeq \sub[\mathcal G]{\an[\mathcal G]{q}}
        \end{align}
        \item (Consistency) Each valuation vertex $x_{v}$ of a visible variable $v \in \mathcal V$ is consistent with the output of the functional parameter $f_{v} \in \mathcal F_{\mathcal V}$ when applied to the valuation vertices $\pa[\mathcal D]{x_{v}}$,
        \begin{align}
            x_{v} = f_{v}(\pa[\mathcal D]{x_{v}})
        \end{align}
        \item (Uniqueness) For each latent variable $\ell \in \mathcal L$, and for every valuation $\lambda_{\ell} \in \Omega_{\ell}$ there exists a \textit{unique} valuation vertex in $\mathcal D$ corresponding to $\lambda_{\ell}$. Unlike latent valuation vertices, the valuations of visible variables $x_{v} \in \Omega_{v}$ may be repeated (or absent) from $\mathcal D$ depending on the form of $\mathcal F_{\mathcal V}$. In such cases, duplicated $x_{v}$'s are always uniquely distinguished by world membership (colored outline).
        \item (Worlds) A \term{world} is a subgraph of $\mathcal D$ that is isomorphic to $\mathcal G$ under the map $\omega_{q} \mapsto q$. Let $\wor{\lambda_{\mathcal L}} \subseteq \mathcal D$ denote the world containing the valuation $\lambda_{\mathcal L} \in \Omega_{\mathcal L}$\footnote{The uniqueness property guarantees that each world $\wor{\lambda_{\mathcal L}}$ is uniquely determined by $\lambda_{\mathcal L}$.}. Furthermore, for any subset $V \subseteq \mathcal V$ of visible variables, let $\obs[V]{\lambda_{\mathcal L}} \in \Omega_{V}$ denote the observed event supported by $\wor{\lambda_{\mathcal L}}$.
        \item (Completeness) For every valuation of the latent variables $\lambda_{\mathcal L} \in \Omega_{\mathcal L}$, there exists a subgraph corresponding to $\wor{\lambda_{\mathcal L}}$.\footnote{Sometimes it is useful to construct an \textit{incomplete} possible worlds diagram; for example, Figure~\ref{fig:bell_pr_box}.}
    \end{enumerate}
\end{definition}

It is important to remark that although a possible worlds diagram $\mathcal D$ can be constructed from the pair ($\mathcal G, \mathcal F_{\mathcal V}$), the two mathematical objects are not equivalent; the functional parameters $\mathcal F_{\mathcal V}$ can contain superfluous information that never appears in $\mathcal D$. We return to this subtle but crucial observation in Section~\ref{sec:symmetry_and_superfluous}.

The essential purpose of the possible worlds construction is as a diagrammatic tool for calculating the observational predictions of a functional causal model. Lemma~\ref{lem:worlds} captures this essence.

\begin{lemma}
    \label{lem:worlds}
    Given a functional causal model $\br{\mathcal G = \br{\mathcal V \cup \mathcal L, \mathcal E}, \mathcal F_{\mathcal V}, \mathcal P_{\mathcal L}}$ (see Definition~\ref{defn:functional_causal_model}), let $\mathcal D$ be the possible worlds diagram for ($\mathcal G, \mathcal F_{\mathcal V}$). The causal compatibility criterion (Equation~\ref{eq:functional_compatibility}) for $\mathcal G$ is equivalent to a probabilistic sum over worlds in $\mathcal D$:
    \begin{align}
        \mathtt{P}_{\mathcal V} = \sum_{\lambda_{\mathcal L} \in \Omega_{\mathcal L}} \prod_{\ell \in \mathcal L} \mathtt{P}_{\ell}(\lambda_{\ell}) [\obs[\mathcal V]{\lambda_{\mathcal L}}]. \label{eq:world_compatibility}
    \end{align}
\end{lemma}

The remainder of this paper explores the consequences of adopting the possible worlds framework as a method for tackling the causal compatibility problem.

\section{A Complete Possibilistic Solution}
\label{sec:complete_possibilistic}

Section~\ref{sec:possible_worlds_framework} introduced the possible worlds framework as a technique for calculating the observable predictions of a functional causal model by means of Lemma~\ref{lem:worlds}. In this section, we use the possible worlds framework to develop a combinatorial algorithm for completely solving the \textit{possibilistic} causal compatibility problem.
\begin{definition}
    Given a probability distribution $\mathtt{P}_{\mathcal V} : \Omega_{\mathcal V} \to \bs{0,1}$, its \term{support} $\sigma(\mathtt{P}_{\mathcal V})$ is defined as the subset of events which are possible,
    \begin{align}
        \sigma(\mathtt{P}_{\mathcal V}) = \bc{x_{\mathcal V} \in \Omega_{\mathcal V} \mid \mathtt{P}_{\mathcal V}(x_{\mathcal V}) > 0}.
    \end{align}
\end{definition}
An observed distribution $\mathtt{P}_{\mathcal V}$ is said to be \textit{possibilistically compatible} with $\mathcal G$ if there exists a functional causal model $\br{\mathcal G, \mathcal F_{\mathcal V}, \mathcal P_{\mathcal L}}$ for which Equation~\ref{eq:functional_compatibility} produces a distribution with the same support as $\mathtt{P}_{\mathcal V}$. The possibilistic variant of the causal compatibility problem is naturally related to the probabilistic causal compatibility problem defined in Definition~\ref{defn:causal_compat}; if a distribution is possibilistically incompatible with $\mathcal G$, then it is also probabilistically incompatible. We now proceed to apply the possible worlds framework to prove possibilistic incompatibility between a number of distribution/causal structure pairs.

\subsection{A Simple Example Causal Structure}
\label{sec:d-separation}

\begin{figure}
    \centering
    \includegraphics{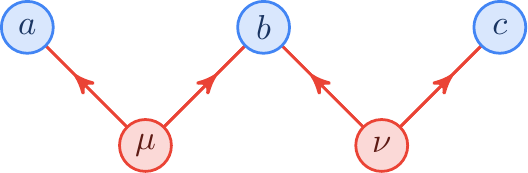}
    \caption{A causal structure $\mathcal G_{\ref{fig:w_structure}}$ with three visible vertices $\mathcal V = \bc{a,b,c}$ and two latent vertices $\mathcal L = \bc{\mu, \nu}$.}
    \label{fig:w_structure}
\end{figure}

\begin{figure}
    \centering
    \begin{subfigure}[t]{0.49\textwidth}
        \centering
        \includegraphics{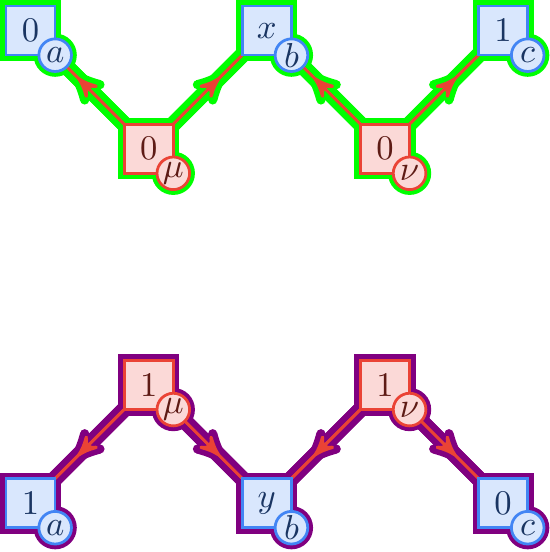}
        \caption{An incomplete possible worlds diagram for $\mathcal G_{\ref{fig:w_structure}}$ initialized by $\mathtt{P}_{abc}^{(\ref{eq:w_dist})}$. The worlds are colored: $\colorline{green}{\wor{0_{\mu}0_{\nu}}}$ green, $\colorline{violet}{\wor{1_{\mu}1_{\nu}}}$ violet.}
        \label{fig:wsdw1}
    \end{subfigure}
    \begin{subfigure}[t]{0.49\textwidth}
        \centering
        \includegraphics{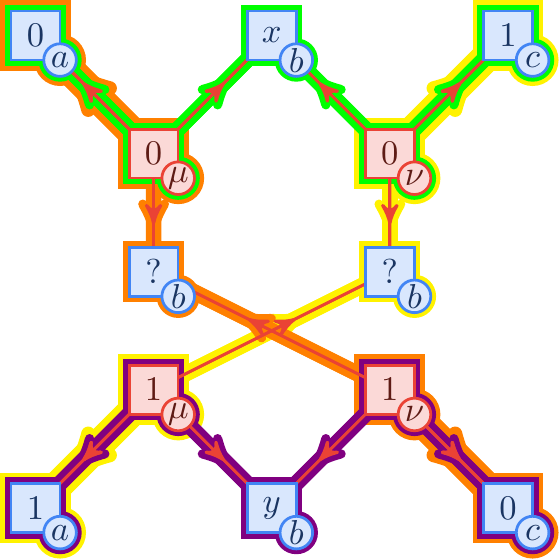}
        \caption{Considering possible worlds produces a contradiction with $\mathtt{P}_{abc}^{(\ref{eq:w_dist})}$. The additional worlds are colored: $\colorline{orange}{\wor{0_{\mu}1_{\nu}}}$ orange, $\colorline{yellow}{\wor{1_{\mu}0_{\nu}}}$ yellow.}
        \label{fig:wsdw2}
    \end{subfigure}
    \caption{The possible worlds diagram for $\mathcal G_{\ref{fig:w_structure}}$ (Figure~\ref{fig:w_structure}) is incompatible with $\mathtt{P}_{abc}^{(\ref{eq:w_dist})}$ (Equation~\ref{eq:w_dist}).}
    \label{fig:w_structure_completion}
\end{figure}
Consider the causal structure $\mathcal G_{\ref{fig:w_structure}}$ depicted in Figure~\ref{fig:w_structure}. For $\mathcal G_{\ref{fig:w_structure}}$, the causal compatibility criteria (Equation~\ref{eq:functional_compatibility}) takes the form,
\begin{align}
    \mathtt{P}_{abc}(x_{a}x_{b}x_{c}) = \sum_{\lambda_{\mu} \in \Omega_{\mu}}\sum_{\lambda_{\nu} \in \Omega_{\nu}} \mathtt{P}_{{\mu}}(\lambda_{\mu})\mathtt{P}_{{\nu}}(\lambda_{\nu}) \delta(x_{a}, f_{a}(\lambda_{\mu}))\delta(x_{b}, f_{b}(\lambda_{\mu}, \lambda_{\nu}))\delta(x_{c}, f_{c}(\lambda_{\nu})). \label{eq:w_compatibility}
\end{align}
The following family of distributions for arbitrary $x_{b},y_{b} \in \Omega_{b}$,
\begin{align}
    \mathtt{P}_{abc}^{(\ref{eq:w_dist})} = z[0_{a}x_{b}1_{c}] + (1-z)[1_{a}y_{b}0_{c}]), \quad 0 < z < 1,
    \label{eq:w_dist}
\end{align}
are incompatible with $\mathcal G_{\ref{fig:w_structure}}$. Traditionally, distributions like $\mathtt{P}_{abc}^{(\ref{eq:w_dist})}$ are proven incompatible on the basis that they violate an independence constraint that is implied by $\mathcal G_{\ref{fig:w_structure}}$~\cite{Pearl_2009}, namely,
\begin{align}
    \forall \mathtt{P}_{abc} \in \mathcal M(\mathcal G_{\ref{fig:w_structure}}), \quad \mathtt{P}_{ac}(x_{a}x_{c}) = \mathtt{P}_{a}(x_{a})\mathtt{P}_{c}(x_{c}).
\end{align}
Intuitively, $\mathcal G_{\ref{fig:w_structure}}$ provides no latent mechanism by which $a$ and $c$ can attempt to correlate (or anti-correlate). We now prove the possibilistic incompatibility of the support $\sigma(\mathtt{P}_{abc}^{(\ref{eq:w_dist})})$ with $\mathcal G_{\ref{fig:w_structure}}$ using the possible worlds framework.
\begin{proof}
    Proof by contradiction; assume that a functional causal model $\mathcal F_{\mathcal V} = \bc{f_{a}, f_{b}, f_{c}}$ for $\mathcal G_{\ref{fig:w_structure}}$ exists such that Equation~\ref{eq:w_compatibility} produces $\mathtt{P}_{abc}^{(\ref{eq:w_dist})}$. Since there are two distinct valuations of the joint variables $abc$ in $\mathtt{P}_{abc}^{(\ref{eq:w_dist})}$, namely $0_{a}x_{b}1_{c}$ and $1_{a}y_{b}0_{c}$, consider each as being sampled from two possible worlds. Without loss of generality\footnote{There is no loss of generality in choosing $0_{\mu}0_{\nu}$ and $1_{\mu}1_{\nu}$ (instead of $0_{\mu}1_{\nu}$ and $1_{\mu}0_{\nu}$) as the valuations for the worlds because the valuation ``{labels}'' associated with latent events are arbitrary. The valuations can not be $0_{\mu}1_{\nu}$ and $1_{\mu}1_{\nu}$ because of the cross-world consistency constraint $\obs[c]{0_{\mu}1_{\nu}} = \obs[c]{1_{\mu}1_{\nu}} = f_{c}(1_{\nu})$.}, let $0_{\mu}0_{\nu} \in \Omega_{\mu} \times \Omega_{\nu}$ denote any valuation of the latent variables such that $\colorline{green}{\obs[abc]{0_{\mu}0_{\nu}}} = 0_{a}x_{b}1_{c}$. Similarly, let $1_{\mu}1_{\nu} \in \Omega_{\mu}\times \Omega_{\nu}$ denote any valuation of the latent variables such that $\colorline{violet}{\obs[abc]{1_{\mu}1_{\nu}}} = 1_{a}y_{b}0_{c}$. Using these observations, initialize a possible worlds diagram using $\colorline{green}{\wor{0_{\mu}0_{\nu}}}$, colored green, and $\colorline{violet}{\wor{1_{\mu}1_{\nu}}}$, colored violet, as seen in Figure~\ref{fig:wsdw1}. In order to complete Figure~\ref{fig:wsdw1}, one simply needs to specify the behavior of $b$ in two of the ``off-diagonal'' worlds, namely $\colorline{orange}{\wor{0_{\mu}1_{\nu}}}$, colored orange, and $\colorline{yellow}{\wor{1_{\mu}0_{\nu}}}$, colored yellow (see Figure~\ref{fig:wsdw2}). Regardless of this choice, the observed event $\colorline{orange}{\obs[ac]{0_{\mu}1_{\nu}} = 0_{a}0_{c}}$ in the orange world $\colorline{orange}{\wor{0_{\mu}1_{\nu}}}$ predicts $\mathtt{P}_{ac}(0_{a}0_{c}) > 0$\footnote{The probabilities associated to each world by Lemma~\ref{lem:worlds} can always be assumed positive, because otherwise, those valuations would be excluded from the latent sample space $\Omega_{\mathcal L}$.} which contradicts $\mathtt{P}_{abc}^{(\ref{eq:w_dist})}$. Therefore, because the proof technique did not rely on the value of $0 < z < 1$, $\mathtt{P}_{abc}^{(\ref{eq:w_dist})}$ is possibilistically incompatible with $\mathcal G_{\ref{fig:w_structure}}$.
\end{proof}

\subsection{The Instrumental Structure}

\begin{figure}
    \centering
    \includegraphics{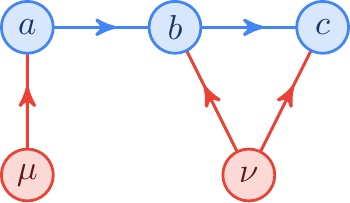}
    \caption{The Instrumental Scenario.}
    \label{fig:instrumental_scenario}
\end{figure}
\begin{figure}
    \centering
    \begin{subfigure}[t]{0.49\textwidth}
        \centering
        \includegraphics{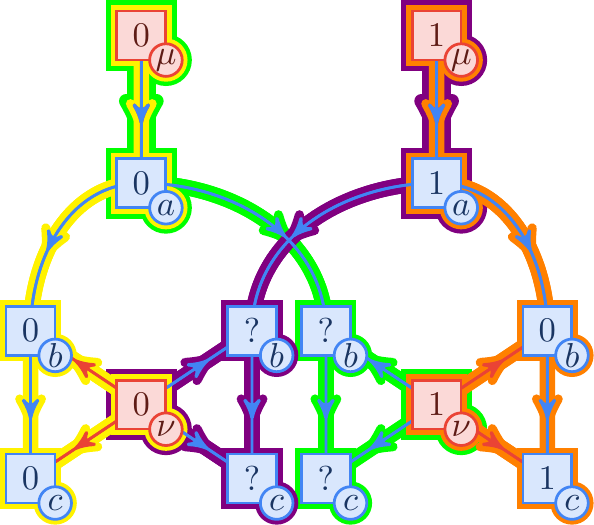}
        \caption{Worlds $\colorline{yellow}{\wor{0_{\mu}0_{\nu}}}$, and $\colorline{orange}{\wor{1_{\mu}1_{\nu}}}$ are initialized by the observed events in Equation~\ref{eq:instrumental_dist}.}
        \label{fig:ispv1}
    \end{subfigure}
    \begin{subfigure}[t]{0.49\textwidth}
        \centering
        \includegraphics{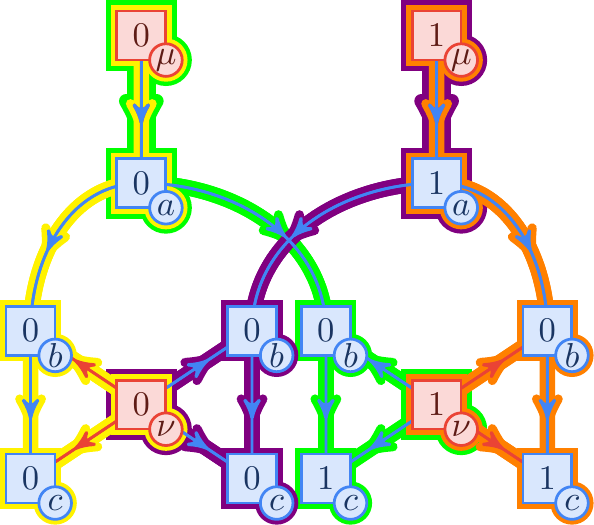}
        \caption{Populating the events in $\colorline{green}{\wor{0_{\mu}1_{\nu}}}$ and $\colorline{violet}{\wor{1_{\mu}0_{\nu}}}$ leads to a contradiction with Equation~\ref{eq:instrumental_dist}. }
        \label{fig:ispv2}
    \end{subfigure}
    \caption{A possible worlds diagram for $\mathcal G_{\ref{fig:instrumental_scenario}}$ (Figure~\ref{fig:instrumental_scenario}). The worlds are colored: $\colorline{yellow}{\wor{0_{\mu}0_{\nu}}}$ yellow, $\colorline{orange}{\wor{1_{\mu}1_{\nu}}}$ orange, $\colorline{violet}{\wor{1_{\mu}0_{\nu}}}$ violet, $\colorline{green}{\wor{0_{\mu}1_{\nu}}}$ green.}
    \label{fig:incomplete_valuation_instrumental}
\end{figure}
The causal structure $\mathcal G_{\ref{fig:instrumental_scenario}}$ depicted in Figure~\ref{fig:instrumental_scenario} is known as the Instrumental Scenario~\cite{Bonet_2013,Pearl_1995,Pearl_2013_2}. For $\mathcal G_{\ref{fig:instrumental_scenario}}$, Equation~\ref{eq:functional_compatibility} takes the form,
\begin{align}
    \mathtt{P}_{{a}{b}{c}}\br{x_{a}x_{b}x_{c}} = \sum_{\lambda_{\mu} \in \Omega_{\mu}}\sum_{\lambda_{\nu} \in \Omega_{\nu}} \mathtt{P}_{{\mu}}(\lambda_{\mu})\mathtt{P}_{{\nu}}(\lambda_{\nu}) \delta(x_{a}, f_{a}(\lambda_{\mu}))\delta(x_{b}, f_{b}(a, \lambda_{\nu}))\delta(x_{c}, f_{c}(b, \lambda_{\nu})). \label{eq:instrumental_compatibility}
\end{align}
The following family of distributions,
\begin{align}
    \label{eq:instrumental_dist}
    \mathtt{P}_{abc}^{(\ref{eq:instrumental_dist})} = z\bs{0_{a}0_{b}0_{c}} + (1-z)\bs{1_{a}0_{b}1_{c}}, \quad 0 < z < 1,
\end{align}
are possibilistically incompatible with $\mathcal G_{\ref{fig:instrumental_scenario}}$. The Instrumental scenario $\mathcal G_{\ref{fig:instrumental_scenario}}$ is different from $\mathcal G_{\ref{fig:w_structure}}$ in that there are no observable conditional independence constraints which can prove the possibilistic incompatibility of $\mathtt{P}_{abc}^{(\ref{eq:instrumental_dist})}$. Instead, the possibilistic incompatibility of $\mathtt{P}_{abc}^{(\ref{eq:instrumental_dist})}$ is traditionally witnessed by an Instrumental inequality originally derived in~\cite{Pearl_1995},
\begin{align}
    \forall \mathtt{P}_{abc} \in \mathcal M(\mathcal G_{\ref{fig:instrumental_scenario}}), \quad \mathtt{P}_{bc|a}(0_{b}0_{c}|0_{a}) + \mathtt{P}_{bc|a}(0_{b}1_{c}|1_{a}) \leq 1.
    \label{eq:instrumental}
\end{align}
Independently of Equation~\ref{eq:instrumental}, we now prove possibilistic incompatibility of $\mathtt{P}_{abc}^{(\ref{eq:instrumental_dist})}$ with $\mathcal G_{\ref{fig:instrumental_scenario}}$ using the possible worlds framework.
\begin{proof}
    Proof by contradiction; assume that a functional model $\mathcal F_{\mathcal V} = \bc{f_{a}, f_{b}, f_{c}}$ for $\mathcal G_{\ref{fig:instrumental_scenario}}$ exists such that Equation~\ref{eq:instrumental_compatibility} produces $\mathtt{P}_{abc}^{(\ref{eq:instrumental_dist})}$ (Equation~\ref{eq:instrumental_dist}). Analogously to the proof in Section~\ref{sec:d-separation}, there are only two distinct valuations of the joint variables $abc$, namely $0_{a}0_{b}0_{c}$ and $1_{a}0_{b}1_{c}$. Therefore, define two worlds one where $\colorline{yellow}{\obs[abc]{0_{\mu}0_{\nu}}} = 0_{a}0_{b}0_{c}$ and another where $\colorline{orange}{\obs[abc]{1_{\mu}1_{\nu}}} = 1_{a}0_{b}1_{c}$. Using these two worlds, a possible worlds diagram can be initialized as in Figure~\ref{fig:ispv1} where $\colorline{yellow}{\wor{0_{\mu}0_{\nu}}}$ is colored yellow and $\colorline{orange}{\wor{1_{\mu}1_{\nu}}}$ is colored orange. In order to complete the possible worlds diagram of Figure~\ref{fig:ispv1}, one first needs to specify how $b$ behaves in two possible worlds: $\colorline{green}{\wor{0_{\mu}1_{\nu}}}$ colored green and $\colorline{violet}{\wor{1_{\mu}0_{\nu}}}$ colored violet.
    \begin{align}
    \begin{split}
        \colorline{violet}{\obs[b]{1_{\mu}0_{\nu}}} &= f_{b}(1_{a}0_{\nu}) = ?_{b}, \\
        \colorline{green}{\obs[b]{0_{\mu}1_{\nu}}} &= f_{b}(0_{a}1_{\nu}) = ?_{b}.
    \end{split}
    \end{align}
    By appealing to $\mathtt{P}_{abc}^{(\ref{eq:instrumental_dist})}$, it must be that $\colorline{violet}{\obs[b]{1_{\mu}0_{\nu}}} = \colorline{green}{\obs[b]{0_{\mu}1_{\nu}}} = 0_{b}$ as no other valuations for $b$ are in the support of $\mathtt{P}_{abc}^{(\ref{eq:instrumental_dist})}$. Finally, the remaining `unknown' observations for $c$ in the violet world $\colorline{violet}{\obs[c]{1_{\mu}0_{\nu}}} = f_{c}(0_b0_{\nu})$, and green world $\colorline{green}{\obs[c]{0_{\mu}1_{\nu}}} = f_{c}(0_b1_{\nu})$ are determined respectively by the behavior of $c$ in the orange $\colorline{orange}{\wor{1_{\mu}1_{\nu}}}$ and yellow $\colorline{yellow}{\wor{0_{\mu}0_{\nu}}}$ worlds as depicted in Figure~\ref{fig:ispv2}. Explicitly,
    \begin{align}
    \label{eq:cross_world_consistency}
    \begin{split}
        \colorline{violet}{\obs[c]{1_{\mu}0_{\nu}}} &= f_{c}(0_b0_{\nu}) = \colorline{yellow}{\obs[c]{0_{\mu}0_{\nu}}} = 0_{c}, \\
        \colorline{green}{\obs[c]{0_{\mu}1_{\nu}}} &= f_{c}(0_b1_{\nu}) = \colorline{orange}{\obs[c]{1_{\mu}1_{\nu}}} = 1_{c}.
    \end{split}
    \end{align}
    Therefore the observed events in the green and violet worlds are fixed to be,
    \begin{align}
        \colorline{green}{\obs[abc]{1_{\mu}0_{\nu}}} = 1_{a}0_{b}0_{c}, \quad \colorline{violet}{\obs[abc]{0_{\mu}1_{\nu}}} = 0_{a}0_{b}1_{c}.
    \end{align}
    Unfortunately, neither of theses events are in the support of $\mathtt{P}_{abc}^{(\ref{eq:instrumental_dist})}$, which is a contradiction; therefore $\mathtt{P}_{abc}^{(\ref{eq:instrumental_dist})}$ is possibilistically incompatible with $\mathcal G_{\ref{fig:instrumental_scenario}}$.
\end{proof}
Notice that unlike the proof from Section~\ref{sec:d-separation}, here we needed to appeal to the cross-world consistency constraints (Equation~\ref{eq:cross_world_consistency}) demanded by the possible worlds framework.

\subsection{The Bell Structure}

\begin{figure}
    \centering
    \includegraphics{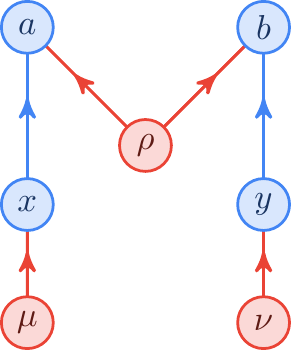}
    \caption{The Bell causal structure has variables $a,b$ `measuring' hidden variable $\rho$ with `measurement settings' $x,y$ determined independently of $\rho$.}
    \label{fig:bell_structure}
\end{figure}
\begin{figure}
    \centering
    \includegraphics{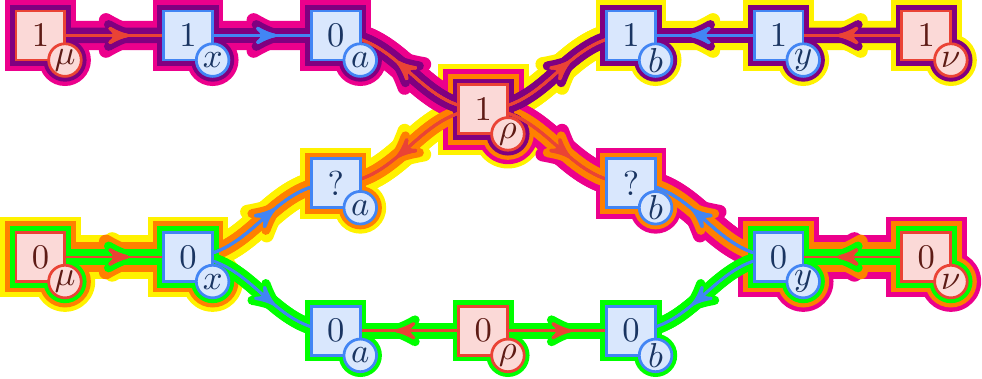}
    \caption{An incomplete possible worlds diagram for the Bell structure $\mathcal G_{\ref{fig:bell_structure}}$ (Figure~\ref{fig:bell_structure}) initialized by the observed events $\colorline{green}{\obs[xaby]{0_{\mu}0_{\rho}0_{\nu}}} = 0_{x}0_{a}0_{b}0_{y}$ and $\colorline{violet}{\obs[xaby]{1_{\mu}1_{\rho}1_{\nu}}} = 1_{x}0_{a}1_{b}1_{y}$. The worlds are colored: $\colorline{green}{\wor{0_{\mu}0_{\rho}0_{\nu}}}$ green, $\colorline{violet}{\wor{1_{\mu}1_{\rho}1_{\nu}}}$ violet, $\colorline{magenta}{\wor{1_{\mu}1_{\rho}0_{\nu}}}$ magenta, $\colorline{yellow}{\wor{0_{\mu}1_{\rho}1_{\nu}}}$ yellow, and $\colorline{orange}{\wor{0_{\mu}1_{\rho}0_{\nu}}}$ orange.}
    \label{fig:bell_pr_box}
\end{figure}
Consider the causal structure $\mathcal G_{\ref{fig:bell_structure}}$ depicted in Figure~\ref{fig:bell_structure} known as the Bell structure~\cite{Bell_1964}. From the perspective of causal inference, Bell's theorem~\cite{Bell_1964} states that any distribution compatible with $\mathcal G_{\ref{fig:bell_structure}}$ must satisfy an inequality constraint known as a Bell inequality. For example, the inequality due to Clauser, Horne, Shimony and Holt, referred to as the CHSH inequality, constrains correlations held between $a$ and $b$ as $x, y$ vary~\cite{CHSH_Original}\footnote{The two variable correlation is defined as $\left\langle ab | x_{x}x_{y}\right\rangle = \sum_{i,j=1}^{2} (-1)^{i + j} \mathtt{P}_{ab|xy}(i_{a}j_{b}|x_{x}x_{y})$. },
\begin{align}
    \forall \mathtt{P}_{xaby} \in \mathcal M(\mathcal G_{\ref{fig:bell_structure}}), \quad S = \left\langle ab|0_{x}0_{y}\right\rangle + \left\langle ab|0_{x}1_{y}\right\rangle + \left\langle ab|1_{x}0_{y}\right\rangle - \left\langle ab|1_{x}1_{y}\right\rangle, \quad \abs{S} \leq 2
    \label{eq:CHSH}
\end{align}
Correlations measured by quantum theory are capable of violating this inequality up to $S = 2 \sqrt{2}$~\cite{Cirelson_1980}. This violation is not maximum; it is possible to achieve a violation of $S = 4$ using Popescu-Rohrlich box correlations~\cite{PR_1995}. The following distribution is an example of a Popescu-Rohrlich box correlation,
\begin{align}
\begin{split}
    \mathtt{P}_{xaby}^{(\ref{eq:bell_dist})}
    = \frac{1}{8}(&[0_{x}0_{a}0_{b}0_{y}] + [0_{x}1_{a}1_{b}0_{y}] + [0_{x}0_{a}0_{b}1_{y}] + [0_{x}1_{a}1_{b}1_{y}] + \\
    + &[1_{x}0_{a}0_{b}0_{y}] + [1_{x}1_{a}1_{b}0_{y}] + [1_{x}0_{a}1_{b}1_{y}] + [1_{x}1_{a}0_{b}1_{y}]).
\end{split}
\label{eq:bell_dist}
\end{align}
Unlike $\mathcal G_{\ref{fig:instrumental_scenario}}$, there are conditional independence constraints placed on correlations compatible with $\mathcal G_{\ref{fig:bell_structure}}$, namely the no-signaling constraints $\mathtt{P}_{a|xy} = \mathtt{P}_{a|x}$ and $\mathtt{P}_{b|xy} = \mathtt{P}_{b|y}$. Because $\mathtt{P}_{xaby}^{(\ref{eq:bell_dist})}$ satisfies the no-signaling constraints, the incompatibility of $\mathtt{P}_{xaby}^{(\ref{eq:bell_dist})}$ with $\mathcal G_{\ref{fig:bell_structure}}$ is traditionally proven using Equation~\ref{eq:CHSH}. We now proceed to prove its incompatibility using the possible worlds framework.
\begin{proof}
    Proof by contradiction; assume that a functional causal model $\mathcal F_{\mathcal V} = \bc{f_{a}, f_{b}, f_{x}, f_{y}}$ for $\mathcal G_{\ref{fig:bell_structure}}$ exists which supports $\mathtt{P}_{xaby}^{(\ref{eq:bell_dist})}$ and use the possible worlds framework. Unlike the previous proofs, we only need to consider a subset of the events in $\mathtt{P}_{xaby}^{(\ref{eq:bell_dist})}$ to initialize a possible worlds diagram. Consider the following pair of events and associated latent valuations which support them\footnote{Clearly, the values of $\lambda_{\mu}$ and $\lambda_{\nu}$ that support these worlds must be unique. Less obvious is the possibility for these worlds to share a $\lambda_{\rho}$ value. Albeit if they do, the event $0_{x}0_{a}1_{b}1_{y}$ becomes possible, contradicting $\mathtt{P}_{xaby}^{(\ref{eq:bell_dist})}$ as well.},
    \begin{align}
        \colorline{green}{\obs[xaby]{0_{\mu}0_{\rho}0_{\nu}}} = 0_{a}0_{b}0_{x}0_{y}, \qquad \colorline{violet}{\obs[xaby]{1_{\mu}1_{\rho}1_{\nu}}} = 1_{a}0_{b}1_{x}1_{y}.
        \label{eq:bell_pr_seed_events}
    \end{align}
    Using Equation~\ref{eq:bell_pr_seed_events}, initialize the possible worlds diagram in Figure~\ref{fig:bell_pr_box} with worlds $\colorline{green}{\wor{0_{\mu}0_{\rho}0_{\nu}}}$ colored green and $\colorline{violet}{\wor{1_{\mu}1_{\rho}1_{\nu}}}$ colored violet. An unavoidable contradiction arises when attempting to populate the values for $f_{a}(0_{x}1_{\rho})$ in the yellow world $\colorline{yellow}{\wor{0_{\mu}1_{\rho}1_{\nu}}}$ and $f_{b}(0_{y}1_{\rho})$ in the magenta world $\colorline{magenta}{\wor{1_{\mu}1_{\rho}0_{\nu}}}$. First, the observed event $\colorline{yellow}{\obs[xaby]{0_{\mu}1_{\rho}1_{\nu}}} = 0_{x}?_{a}1_{b}1_{y}$ in the yellow world $\colorline{yellow}{\wor{0_{\mu}1_{\rho}1_{\nu}}}$ must belong to the list of possible events prescribed by $\mathtt{P}_{xaby}^{(\ref{eq:bell_dist})}$; a quick inspection leads one to recognize that the only possibility is $\colorline{yellow}{\obs[a]{0_{\mu}1_{\rho}1_{\nu}}} = f_{a}(0_{x}1_{\rho}) = 1_{a}$. An analogous argument in the magenta world $\colorline{magenta}{\wor{1_{\mu}1_{\rho}0_{\nu}}}$ proves that $\colorline{magenta}{\obs[b]{1_{\mu}1_{\rho}0_{\nu}}} = f_{b}(0_{y}1_{\rho}) = 0_{b}$. Therefore, the observed event in the orange world $\colorline{orange}{\wor{0_{\mu}1_{\rho}0_{\nu}}}$ must be,
    \begin{align}
        \colorline{orange}{\obs[abcd]{0_{\mu}1_{\rho}0_{\nu}}} = 0_{x}1_{a}0_{b}0_{y},
    \end{align}
    and therefore $\mathtt{P}_{xaby}(0_{x}1_{a}0_{b}0_{y}) > 0$ which contradicts $\mathtt{P}_{xaby}^{(\ref{eq:bell_dist})}$. Therefore, $\mathtt{P}_{xaby}^{(\ref{eq:bell_dist})}$ is possibilistically\footnote{The proof holds if the probabilities of the events in $\mathtt{P}_{xaby}^{(\ref{eq:bell_dist})}$ are any positive value.} incompatible with $\mathcal G_{\ref{fig:bell_structure}}$.
\end{proof}

\subsection{The Triangle Structure}

\begin{figure}
    \centering
    \includegraphics{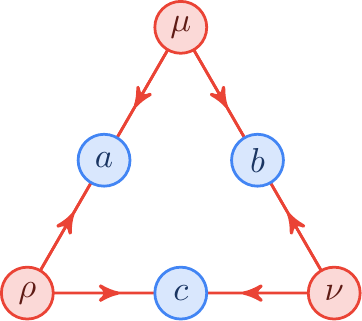}
    \caption{The Triangle structure $\mathcal G_{\ref{fig:triangle_structure}}$ involving three visible variables $\mathcal V = \bc{a, b, c}$ each sharing a pair of latent variables from $\mathcal L = \bc{\mu, \nu, \rho}$.}
    \label{fig:triangle_structure}
\end{figure}
\begin{figure}
    \centering
    \includegraphics{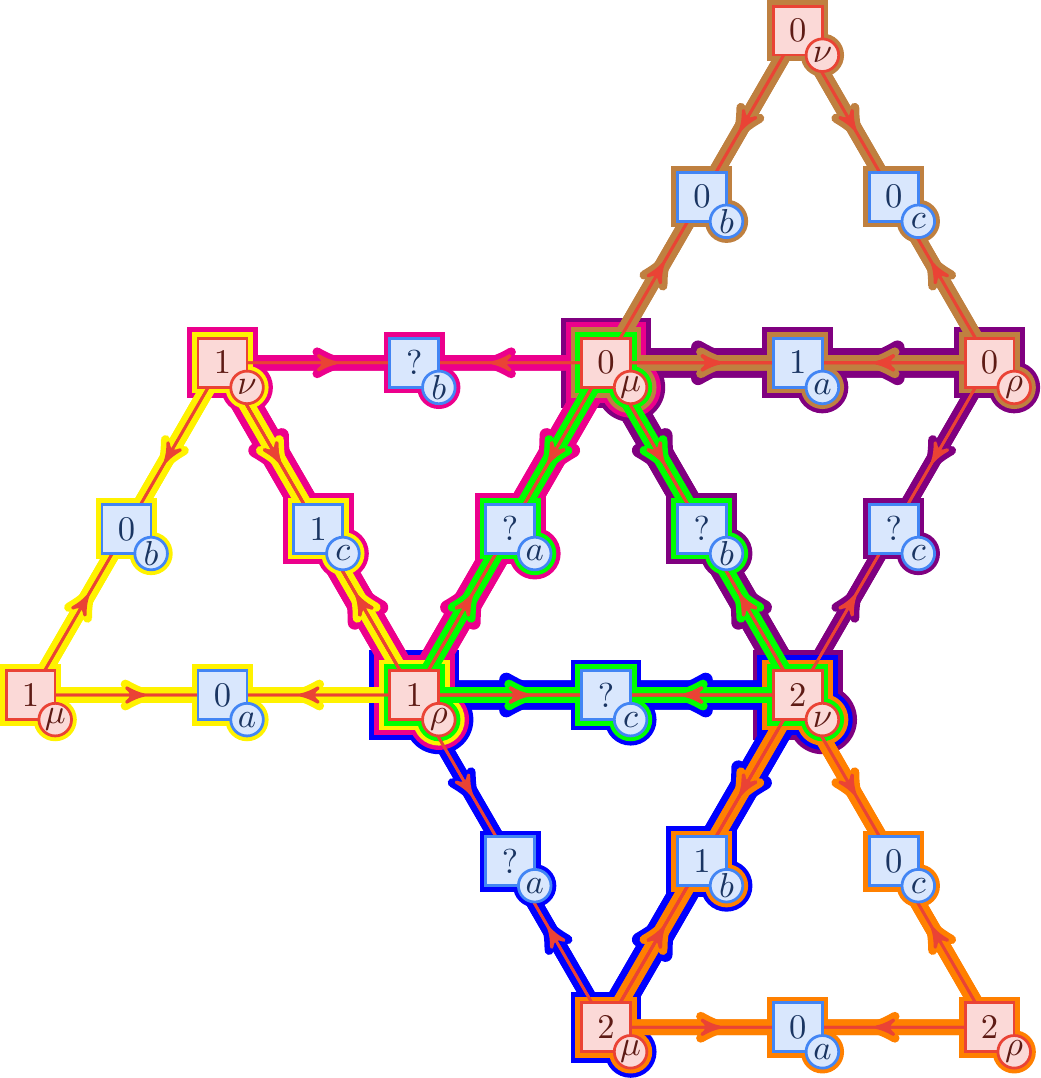}
    \caption{An incomplete possible worlds diagram for the Triangle structure $\mathcal G_{\ref{fig:triangle_structure}}$ (Figure~\ref{fig:triangle_structure}) initialized by the triplet of observed events in Equation~\ref{eq:triangle_seed_events}. The worlds are colored: $\colorline{brown}{\wor{0_{\mu}0_{\nu}0_{\rho}}}$ brown, $\colorline{yellow}{\wor{1_{\mu}1_{\nu}1_{\rho}}}$ yellow, $\colorline{orange}{\wor{2_{\mu}2_{\nu}2_{\rho}}}$ orange, $\colorline{magenta}{\wor{0_{\mu}1_{\nu}1_{\rho}}}$ magenta, $\colorline{blue}{\wor{2_{\mu}2_{\nu}1_{\rho}}}$ blue, $\colorline{violet}{\wor{0_{\mu}2_{\nu}0_{\rho}}}$ violet, and $\colorline{green}{\wor{0_{\mu}2_{\nu}1_{\rho}}}$ green.}
    \label{fig:triangle_structure_definite_worlds}
\end{figure}
Consider the causal structure $\mathcal G_{\ref{fig:triangle_structure}}$ depicted in Figure~\ref{fig:triangle_structure} known as the Triangle structure. The Triangle has been studied extensively in recent decades~\cite{Studel_2015,Fritz_2012,Chaves_2014,Branciard_2012,Henson_2014,Weilenmann_2016,Navascues_2017,Wolfe_2016,Fraser_2017}. The following family of distributions are possibilistically incompatible with $\mathcal G_{\ref{fig:triangle_structure}}$\footnote{The Inflation Technique first proved the incompatibility between $\mathtt{P}_{abc}^{(\ref{eq:triangle_dist})}$ and $\mathcal G_{\ref{fig:triangle_structure}}$.},
\begin{align}
\begin{split}
    \mathtt{P}_{abc}^{(\ref{eq:triangle_dist})}
    = p_{1}[1_{a}0_{b}0_{c}] + p_{2}[0_{a}1_{b}0_{c}] + p_{3}[0_{a}0_{b}1_{c}], \quad \sum_{i=1}^{3} p_{i} = 1, p_{i} > 0.
\end{split}
\label{eq:triangle_dist}
\end{align}
\begin{proof}
    Proof by contradiction: assume that a functional causal model $\mathcal F_{\mathcal V} = \bc{f_{a}, f_{b}, f_{c}}$ for $\mathcal G_{\ref{fig:triangle_structure}}$ exists supporting $\mathtt{P}_{abc}^{(\ref{eq:triangle_dist})}$ and use the possible worlds framework. For each distinct event in $\mathtt{P}_{abc}^{(\ref{eq:triangle_dist})}$, consider a world in which it happens definitely. Explicitly define,
    \begin{align}
         \colorline{brown}{\obs[abc]{0_{\mu}0_{\rho}0_{\nu}}} = 1_{a}0_{b}0_{c}, \\
        \colorline{yellow}{\obs[abc]{1_{\mu}1_{\rho}1_{\nu}}} = 0_{a}0_{b}1_{c}, \\
        \colorline{orange}{\obs[abc]{2_{\mu}2_{\rho}2_{\nu}}} = 0_{a}1_{b}0_{c},
        \label{eq:triangle_seed_events}
    \end{align}
    corresponding to the exterior worlds in Figure~\ref{fig:triangle_structure_definite_worlds}. Consider magenta world $\colorline{magenta}{\wor{0_{\mu}1_{\rho}1_{\nu}}}$ with partially specified observation $\colorline{magenta}{\obs[abc]{0_{\mu}1_{\rho}1_{\nu}}} = ?_{a}?_{b}1_{c}$. Recalling $\mathtt{P}_{abc}^{(\ref{eq:triangle_dist})}$, whenever $c$ takes value $1_{c}$, \textit{both} $a$ and $b$ take the value $0$; i.e. $0_{a}0_{b}$. Therefore, it must be that the observed event in the magenta world $\colorline{magenta}{\wor{0_{\mu}1_{\rho}1_{\nu}}}$ is $\colorline{magenta}{\obs[abc]{0_{\mu}1_{\rho}1_{\nu}}} = 0_{a}0_{b}1_{c}$. An analogous argument holds for other worlds,
    \begin{align}
    \label{eq:triangle_false_conclusions}
    \begin{split}
        \colorline{magenta}{\obs[abc]{0_{\mu}1_{\rho}1_{\nu}}} = ?_{a}?_{b}1_{c} &\Rightarrow \colorline{magenta}{\obs[abc]{0_{\mu}1_{\rho}1_{\nu}}} = 0_{a}0_{b}1_{c}, \\
        \colorline{blue}{\obs[abc]{2_{\mu}2_{\rho}1_{\nu}}} = ?_{a}1_{b}?_{c} &\Rightarrow \colorline{blue}{\obs[abc]{2_{\mu}2_{\rho}1_{\nu}}} = 0_{a}1_{b}0_{c}, \\
        \colorline{violet}{\obs[abc]{0_{\mu}2_{\rho}0_{\nu}}} = 1_{a}?_{b}?_{c} &\Rightarrow \colorline{violet}{\obs[abc]{0_{\mu}2_{\rho}0_{\nu}}} = 1_{a}0_{b}0_{c}.
    \end{split}
    \end{align}
    However, the conclusions drawn by Equation~\ref{eq:triangle_false_conclusions} predict the observed event the in central, green world $\colorline{green}{\wor{0_{\mu}2_{\rho}1_{\nu}}}$ must be,
    \begin{align}
        \colorline{green}{\obs[abc]{0_{\mu}2_{\rho}1_{\nu}}} = 0_{a}0_{b}0_{c},
    \end{align}
    and therefore $\mathtt{P}_{abc}(0_{a}0_{b}0_{c}) > 0$ which contradicts $\mathtt{P}_{abc}^{(\ref{eq:triangle_dist})}$. Therefore, $\mathtt{P}_{abc}^{(\ref{eq:triangle_dist})}$ is possibilistically incompatible with $\mathcal G_{\ref{fig:triangle_structure}}$.
\end{proof}

\subsection{An Evans Causal Structure}

\begin{figure}
    \centering
    \includegraphics{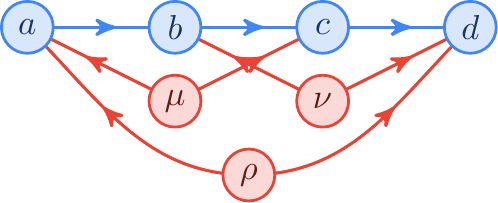}
    \caption{The Evans Causal Structure $\mathcal G_{\ref{fig:evans_causal_structure_1}}$.}
    \label{fig:evans_causal_structure_1}
\end{figure}
Consider the causal structure in Figure~\ref{fig:evans_causal_structure_1}, denoted $\mathcal G_{\ref{fig:evans_causal_structure_1}}$. This causal structure was first mentioned by Evans~\cite{Evans_2016}, along with two others, as one for which no existing techniques were able to prove whether or not it was saturated; that is, whether or not \textit{all} distributions were compatible with it. Here it is shown that there are indeed distributions which are possibilistically \textit{incompatible} with $\mathcal G_{\ref{fig:evans_causal_structure_1}}$ using the framework of possible worlds diagrams. As such, this framework currently stands as the most powerful method for deciding possibilistic compatibility.

Consider the family of distributions with three possible events:
\begin{align}
    \mathtt{P}_{abcd}^{(\ref{eq:evans_failing_distribution})} = p_{1} [0_{a}0_{b}0_{c}y_{d}] + p_{2} [1_{a}0_{b}1_{c}0_{d}] + p_{3} [0_{a}1_{b}1_{c}1_{d}], \quad \sum_{i=1}^{3} p_{i} = 1, p_{i} > 0.
    \label{eq:evans_failing_distribution}
\end{align}
Regardless of the values for $p_{1},p_{2},p_{3}$ (and $y_{d} \in \Omega_{d}$ arbitrary), $\mathtt{P}_{abcd}^{(\ref{eq:evans_failing_distribution})}$ is incompatible with $\mathcal G_{\ref{fig:evans_causal_structure_1}}$.

\begin{figure}
    \centering
    \includegraphics{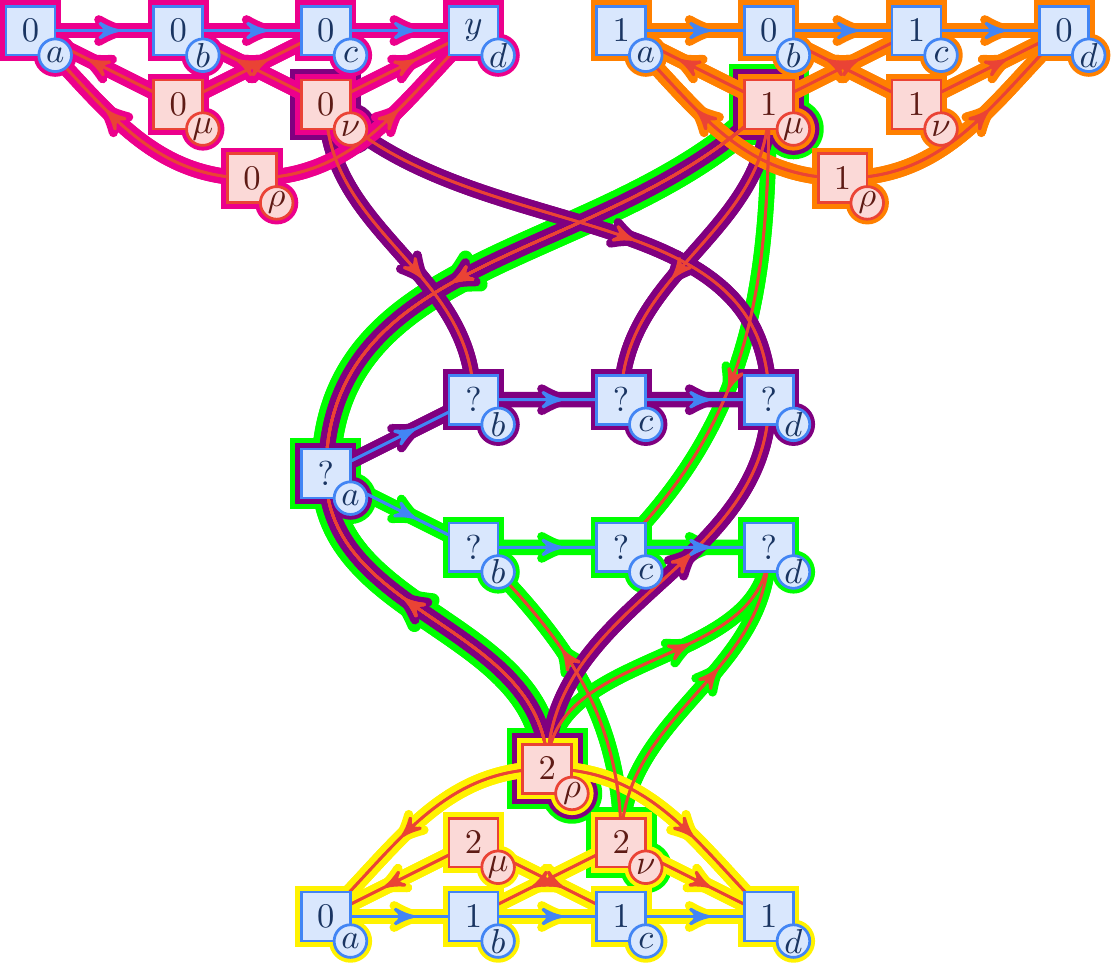}
    \caption{A possible worlds diagram for $\mathcal G_{\ref{fig:evans_causal_structure_1}}$ initialized by the distribution in Equation~\ref{eq:evans_failing_distribution}. The worlds are colored: $\colorline{magenta}{\wor{0_{\mu}0_{\nu}0_{\rho}}}$ magenta, $\colorline{orange}{\wor{1_{\mu}1_{\nu}1_{\rho}}}$ orange, $\colorline{yellow}{\wor{2_{\mu}2_{\nu}2_{\rho}}}$ yellow, $\colorline{violet}{\wor{1_{\mu}0_{\nu}2_{\rho}}}$ violet, and $\colorline{green}{\wor{1_{\mu0}2_{\nu}2_{\rho}}}$ green. }
    \label{eq:evans_definite_worlds}
\end{figure}
\begin{proof}
    Proof by contradiction. First assume that a deterministic model $\mathcal F_{\mathcal V} = \bc{f_{a}, f_{b}, f_{c}, f_{d}}$ for $\mathtt{P}_{abcd}^{(\ref{eq:evans_failing_distribution})}$ exists and adopt the possible worlds framework. Let $\wor{i_{\mu}i_{\nu}i_{\rho}}$ for $i\in \bc{1,2,3}$ index the possible worlds which support the events observed in $\mathtt{P}_{abcd}$,
    \begin{align}
    \begin{split}
        \colorline{magenta}{\obs[abcd]{0_{\mu}0_{\nu}0_{\rho}}} = 0_{a}0_{b}0_{c}y_{d}, \\
        \colorline{orange}{\obs[abcd]{1_{\mu}1_{\nu}1_{\rho}}} = 1_{a}0_{b}1_{c}0_{d}, \\
        \colorline{yellow}{\obs[abcd]{2_{\mu}2_{\nu}2_{\rho}}} = 0_{a}1_{b}1_{c}1_{d}.
    \end{split}
    \end{align}
    Only two additional possible worlds are necessary for achieving a contradiction. Consulting Figure~\ref{eq:evans_definite_worlds} for details, these possible worlds are $\colorline{violet}{\wor{1_{\mu}0_{\nu}2_{\rho}}}$ colored violet and $\colorline{green}{\wor{1_{\mu}2_{\nu}2_{\rho}}}$ colored green.
    Notice that the determined value for $a$ must be the same in both worlds as it is independent of $\lambda_{\nu}$:
    \begin{align}
        x_{a} = f_{a}(1_{\mu}2_{\rho}) = \colorline{violet}{\obs[a]{1_{\mu}0_{\nu}2_{\rho}}} = \colorline{green}{\obs[a]{1_{\mu}2_{\nu}2_{\rho}}}.
    \end{align}
    There are only two possible values for $x_{a}$ in any world, namely $x_{a} = 0_{a}$ or $x_{a} = 1_{a}$ as given by $\mathtt{P}_{abcd}^{(\ref{eq:evans_failing_distribution})}$. First suppose that $x_{a} = 0_{a}$. Then in the violet world $\colorline{violet}{\wor{1_{\mu}0_{\nu}2_{\rho}}}$, the value of $b$, to be $\colorline{violet}{\obs[b]{1_{\mu}0_{\nu}2_{\rho}}} = f_{b}(0_{a}0_{\nu}) = 0_{b}$ is completely constrained by consistency with the magenta world $\colorline{magenta}{\wor{0_{\mu}0_{\nu}0_{\rho}}}$. Therefore, $\colorline{violet}{\obs[ab]{1_{\mu}0_{\nu}2_{\rho}}} = 0_{a}0_{b}$. By analogous logic, in the violet world the value of $c$ is constrained to be $\colorline{violet}{\obs[c]{1_{\mu}0_{\nu}2_{\rho}}} = f_{c}(0_{b}1_{\mu}) = 0_{c}$ by the orange world $\colorline{orange}{\wor{1_{\mu}1_{\nu}1_{\rho}}}$. Therefore, $\colorline{violet}{\obs[abc]{1_{\mu}0_{\nu}2_{\rho}}} = 0_{a}0_{b}0_{c}$, which is a contradiction because $0_{a}0_{b}0_{c}$ is an impossible event in $\mathtt{P}_{abcd}^{(\ref{eq:evans_failing_distribution})}$. Therefore, it must be that $x_{a} = 1_{a}$. An unavoidable contradiction follows from attempting to populate the green world $\colorline{green}{\wor{1_{\mu}2_{\nu}2_{\rho}}}$ in Figure~\ref{eq:evans_definite_worlds} with the established knowledge that $\colorline{green}{\obs[a]{1_{\mu}2_{\nu}2_{\rho}}} = 1_{a}$. The value of $\colorline{green}{\obs[b]{1_{\mu}2_{\nu}2_{\rho}}} = f_{b}(1_{a}1_{\nu})$ has yet to be specified by any possible worlds, but choosing $f_{b}(1_{a}1_{\nu}) = 1_{b}$ would yield an impossible event $\colorline{green}{\obs[a]{1_{\mu}2_{\nu}2_{\rho}}} = 1_{a}1_{b}$. Therefore, it must be that $f_{b}(1_{a}1_{\nu}) = 0_{b}$ and $\colorline{green}{\obs[a]{1_{\mu}2_{\nu}2_{\rho}}} = 1_{a}0_{b}$. Similarly, the orange world $\colorline{orange}{\wor{1_{\mu}1_{\nu}1_{\rho}}}$ fixes $f_{c}(0_{b}1_{\mu}) = 1_{c}$ and therefore $\colorline{green}{\obs[abc]{1_{\mu}2_{\nu}2_{\rho}}} = 1_{a}0_{b}1_{c}$. Finally, the yellow world $\colorline{yellow}{\wor{2_{\mu}2_{\nu}2_{\rho}}}$ already determines $\colorline{green}{\obs[d]{1_{\mu}2_{\nu}2_{\rho}}} = f_{d}(0_{c}2_{\nu}2_{\rho}) = 1_{d}$ and therefore one concludes that,
    \begin{align}
        \colorline{green}{\obs[abcd]{1_{\mu}2_{\nu}2_{\rho}}} = 1_{a}0_{b}1_{c}1_{d},
    \end{align}
    which is an impossible event in $\mathtt{P}_{abcd}^{(\ref{eq:evans_failing_distribution})}$. This contradiction implies that \textit{no} functional model $\mathcal F_{\mathcal V} = \bc{f_{a}, f_{b}, f_{c}, f_{d}}$ exists and therefore $\mathtt{P}_{abcd}^{(\ref{eq:evans_failing_distribution})}$ is possibilistically incompatible with $\mathcal G_{\ref{fig:evans_causal_structure_1}}$.
\end{proof}

To reiterate, there are currently no other methods known~\cite{Evans_2016} which are capable of proving the incompatibility of any distribution with $\mathcal G_{\ref{fig:evans_causal_structure_1}}$\footnote{It is worth noting we have also proven the non-saturation of the other two causal structures mention in~\cite{Evans_2016} using analogous proofs.}. Therefore, the possible worlds framework can be seen as the state-of-the-art technique for determining possibilistic causation.

\subsection{Necessity and Sufficiency}
Throughout this section, we explored a number of proofs of possibilistic incompatibility using the possible worlds framework. Moreover, the above examples communicate a systematic algorithm for deciding possibilistic compatibility. Given a distribution $\mathtt{P}_{\mathcal V}$ with support $\sigma(\mathtt{P}_{\mathcal V}) \subset \Omega_{\mathcal V}$, and a causal structure $\mathcal G = \br{\mathcal V \cup \mathcal L, \mathcal E}$, the following algorithm sketch determines if $\mathtt{P}_{\mathcal V}$ is possibilistically compatible with $\mathcal G$.

\begin{enumerate}
    \item Let $W = \abs{\sigma(\mathtt{P}_{\mathcal V})} < \abs{\Omega_{\mathcal V}}$ denote the number of possible events provided by $\mathtt{P}_{\mathcal V}$.
    \item For each $1 \leq i \leq W$, create a possible world $\wor{\lambda^{(i)}_{\mathcal L}}$ where $\lambda^{(i)}_{\mathcal L} = \bc{i_{\ell} \mid \ell \in \mathcal L}$, thus defining the latent sample space $\Omega_{\mathcal L}$.
    \item Attempt to complete the possible worlds diagram $\mathcal D$ initialized by the worlds $\bc{\wor{\lambda^{(i)}_{\mathcal L}}}_{i=1}^{W}$.
    \item If an impossible event $x_{\mathcal V} \not \in \sigma(\mathtt{P}_{\mathcal V})$ is produced by any ``off-diagonal'' world $\wor{\ldots i_{\ell} \ldots j_{\ell'}\ldots }$ where $i \neq j$, or if a cross-world consistency constraint is broken, back-track.
\end{enumerate}
Upon completing the search, there are two possibilities. The first possibility is that the algorithm returns a completed, consistent, possible worlds diagram $\mathcal D$. Then by Lemma~\ref{lem:worlds}, $\mathtt{P}_{\mathcal V}$ is possibilistically compatible with $\mathcal G$. The second possibility is that an unavoidable contradiction arises, and $\mathtt{P}_{\mathcal V}$ is \textit{not} possibilistically compatible with $\mathcal G$.\footnote{A simple C implementation of the above pseudo-algorithm for boolean visible variables ($|\Omega_{v}| = 2, \forall v \in \mathcal V$) can be found at \href{https://github.com/tcfraser/possibilistic_causality}{github.com/tcfraser/possibilistic\_causality}. In particular, the provided software can output a DIMACS formatted CNF file for usage in most popular boolean satisfiability solvers.}

\section{A Complete Probabilistic Solution}
\label{sec:complete_probabilistic}

In Section~\ref{sec:complete_possibilistic}, we demonstrated that the possible worlds framework was capable of providing a complete possibilistic solution to the causal compatibility problem. If however, a given distribution $\mathtt{P}_{\mathcal V}$ happens to satisfy a causal hypothesis on a possibilistic level, can the possible worlds framework be used to determine if $\mathtt{P}_{\mathcal V}$ satisfies the causal hypothesis on a \textit{probabilistic} level as well? In this section, we answer this question affirmatively. In particular, we provide a hierarchy of feasibility tests for probabilistic compatibility which converges exactly. In addition, we illustrate that a possible worlds diagram is the natural data structure for algorithmically implementing this converging hierarchy.

\subsection{Symmetry and Superfluity}
\label{sec:symmetry_and_superfluous}

This aforementioned hierarchy of tests, to be explained in Section~\ref{sec:a_converging_hierarchy_of_tests}, relies on the enumeration of all probability distributions $\mathtt{P}_{\mathcal V}$ which admit \term{uniform} functional causal models $\br{\mathcal G, \mathcal F_{\mathcal V}, \mathcal P_{\mathcal L}}$ for fixed cardinalities $k_{\mathcal V \cup \mathcal L} = \bc{k_{q} = \abs{\Omega_{q}} \mid q \in \mathcal V \cup \mathcal L}$. A functional causal model is \textit{uniform} if the probability distributions $\mathtt{P}_{\ell} \in \mathcal P_{\mathcal L}$ over the latent variables are uniform distributions; $\mathtt{P}_{\ell} : \Omega_{\ell} \to k_{\ell}^{-1}$. Section~\ref{sec:uniformity} discusses why uniform functional causal models are worth considering, whereas in this section, we discuss how to efficiently enumerate \textit{all} probability distributions $\mathtt{P}_{\mathcal V}$ that are uniformly generated from fixed cardinalities $k_{\mathcal V \cup \mathcal L}$.

One method for generating all such distributions is to perform a brute force enumeration of all deterministic strategies $\mathcal F_{\mathcal V}$ for fixed cardinalities $k_{\mathcal V \cup \mathcal L}$. Depending on the details of the causal structure, the number of deterministic functions of this form is poly-exponential in the cardinalities $k_{\mathcal V \cup \mathcal L}$. This method is inefficient because is fails to consider that many distinct deterministic strategies produce the exact same distribution $\mathtt{P}_{\mathcal V}$. There are two optimizations that can be made to avoid regenerations of the same distribution $\mathtt{P}_{\mathcal V}$ while enumerating all deterministic strategies $\mathcal F_{\mathcal V}$. These optimizations are best motivated by an example using the possible worlds framework.

Consider the causal structure $\mathcal G_{\ref{fig:sym_structure}}$ in Figure~\ref{fig:sym_structure} with visible variables $\mathcal V = \bc{a,b,c}$ and latent variables $\mathcal L = \bc{\mu, \nu}$. Furthermore, for concreteness, suppose that $k_{\mu} = k_{\nu} = k_{a} = k_{a} = 2$ and $k_{c} = 4$. Finally let $\mathcal F_{\mathcal V} = \bc{f_{a}, f_{b}, f_{c}}$ be such that,
\begin{align}
\begin{split}
    f_{a}(0_{\mu}) = 0_{a}, \quad f_{a}(1_{\mu}) = 1_{a}, &\quad
    f_{b}(0_{\mu}) = 0_{b}, \quad f_{b}(1_{\mu}) = 1_{b}, \\
    f_{c}(0_{a}0_{b}0_{\nu}) = 2_{c}, \quad f_{c}(0_{a}0_{b}1_{\nu}) = 0_{c}, &\quad f_{c}(1_{a}1_{b}0_{\nu}) = 3_{c}, \quad f_{c}(1_{a}1_{b}1_{\nu}) = 1_{c} \\
    f_{c}(0_{a}1_{b}0_{\nu}) = 0_{c}, \quad f_{c}(0_{a}1_{b}1_{\nu}) = 1_{c}, &\quad f_{c}(1_{a}0_{b}0_{\nu}) = 2_{c}, \quad f_{c}(1_{a}0_{b}1_{\nu}) = 3_{c}.
\end{split}
\label{eq:example_functional_dependences_2}
\end{align}
The possible worlds diagram $\mathcal D$ for $\mathcal G_{\ref{fig:sym_structure}}$ generated by Equation~\ref{eq:example_functional_dependences_2} is depicted in Figure~\ref{fig:sym_no_extra}. If the latent valuations are distributed uniformly, the probability distribution associated with Figure~\ref{fig:sym_no_extra} (as given by Equation~\ref{eq:world_compatibility}) is equal to,
\begin{align}
\begin{split}
    \mathtt{P}_{abc}
    &= \frac{1}{4}( [\colorline{green}{\wor{0_{\mu}0_{\nu}}}] + [\colorline{orange}{\wor{0_{\mu}1_{\nu}}}] + [\colorline{yellow}{\wor{1_{\mu}0_{\nu}}}] + [\colorline{violet}{\wor{1_{\mu}1_{\nu}}}] ) \\
    &= \frac{1}{4}( [0_{a}0_{b}2_{c}] + [0_{a}0_{b}0_{c}] + [1_{a}1_{b}3_{c}] + [1_{a}1_{b}1_{c}] ).
\end{split}
\label{eq:example_distribution_2}
\end{align}
The first optimization comes from noticing that Equation~\ref{eq:example_functional_dependences_2} specifies how $c$ would respond if provided with the valuation $1_{a}0_{b}1_{\nu}$ of its parents, namely $f_{c}(1_{a}0_{b}1_{\nu}) = 3_{c}$. Nonetheless, this hypothetical scenario is excluded from Figure~\ref{fig:sym_no_extra} (crossed out in the figure) because the functional model in Equation~\ref{eq:example_functional_dependences_2} never produces an opportunity for $a$ to be different from $b$. Consequently, the functional dependences in Equation~\ref{eq:example_functional_dependences_2} contain \textit{superfluous} information irrelevant to the observed probability distribution in Equation~\ref{eq:example_distribution_2}.

Therefore, a brute force enumeration of deterministic strategies would regenerate Equation~\ref{eq:example_distribution_2} several times, once for each assignment of $c$'s behavior in these superfluous scenarios. It is possible to avoid these regenerations by using an unpopulated possible worlds diagram $\tilde {\mathcal D}$ as a data structure and performing a brute force enumeration of all \textit{consistent} valuations of $\tilde {\mathcal D}$.

The second optimization comes from noticing that Equation~\ref{eq:example_distribution_2} contains many \textit{symmetries}. Notably, independently permuting the latent valuations, $\pi_{\mu} : 0_{\mu} \leftrightarrow 1_{\mu}$ or $\pi_{\nu} : 0_{\nu} \leftrightarrow 1_{\nu}$, leaves the observed distribution in Equation~\ref{eq:example_distribution_2} invariant, but maps the functional dependences $\mathcal F_{\mathcal V}$ of Equation~\ref{eq:example_functional_dependences_2} to different functional dependences $\mathcal F_{\mathcal V}^{\pi_{\mu}}$ and $\mathcal F_{\mathcal V}^{\pi_{\nu}}$. These symmetries are reflected as permutations of the worlds as depicted in Figures~\ref{fig:sym_1}, and~\ref{fig:sym_2}.

Analogously, it is possible to avoid these regenerations by first pre-computing the induced action on $\tilde {\mathcal D}$, and thus an induced action on $\mathcal F_{\mathcal V}$, under the permutation group $S_{\mathcal L} = \prod_{\ell \in \mathcal L} \perm{\Omega_{\ell}}$. Then, using the permutation group $S_{\mathcal L}$, one only needs to generate a representative from the equivalence classes of possible worlds diagrams $\mathcal D$ under $S_{\mathcal L}$.

\begin{figure}
    \centering

    \begin{subfigure}[t]{0.9\textwidth}
        \centering
        \includegraphics{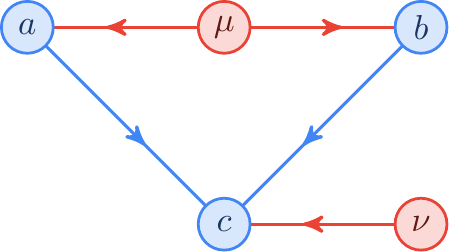}
        \caption{A causal structure $\mathcal G_{\ref{fig:sym_structure}}$ with three visible variables $\mathcal V = \bc{a,b,c}$ and two latent variables $\mathcal L = \bc{\mu, \nu}$.}
        \label{fig:sym_structure}
    \end{subfigure}
    \begin{subfigure}[t]{0.9\textwidth}
        \centering
        \includegraphics{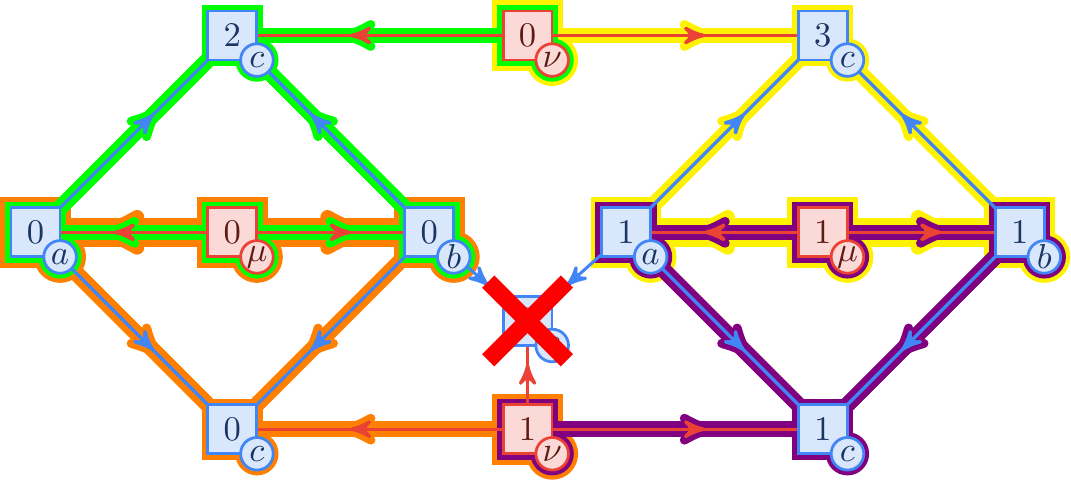}
        \caption{A possible worlds diagram for $\mathcal G_{\ref{fig:sym_structure}}$. The crossed out vertex is excluded because it fails to satisfy the ancestral isomorphism property.}
        \label{fig:sym_no_extra}
    \end{subfigure}
    \begin{subfigure}[t]{0.9\textwidth}
        \centering
        \includegraphics{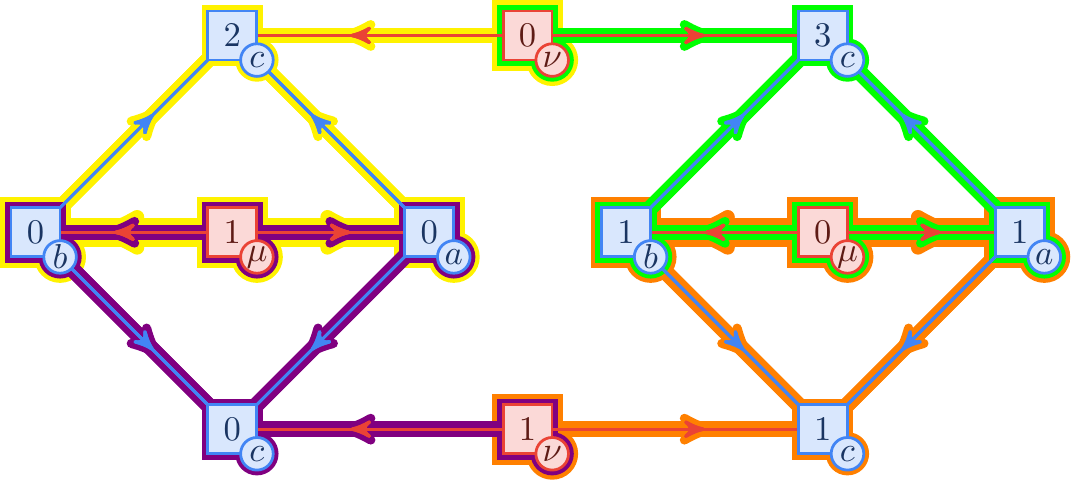}
        \caption{The image of Figure~\ref{fig:sym_no_extra} under the permutation $0_{\mu} \leftrightarrow 1_{\mu}$.}
        \label{fig:sym_1}
    \end{subfigure}
    \begin{subfigure}[t]{0.9\textwidth}
        \centering
        \includegraphics{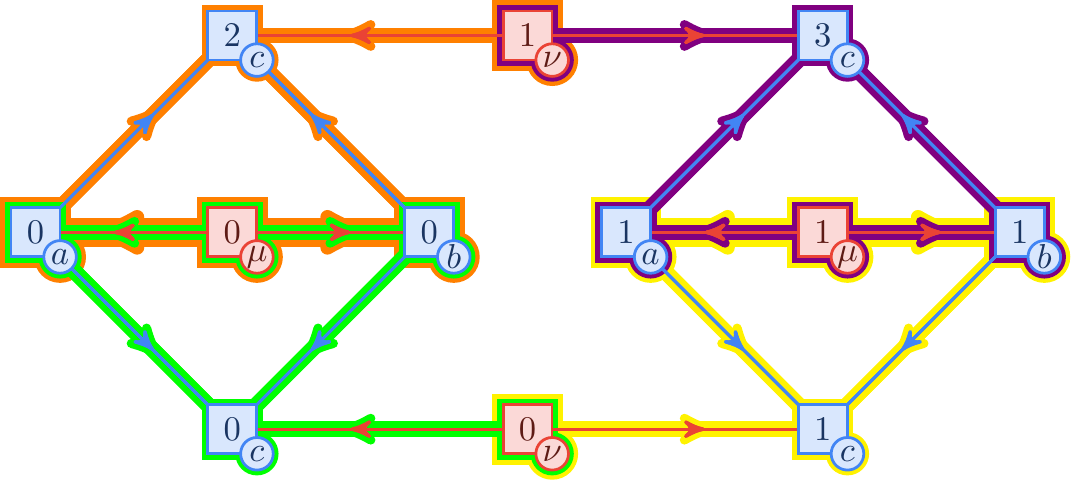}
        \caption{The image of Figure~\ref{fig:sym_no_extra} under the permutation $0_{\nu} \leftrightarrow 1_{\nu}$.}
        \label{fig:sym_2}
    \end{subfigure}
    \caption{Every permutation $\pi_{\ell} : \Omega_{\ell} \to \Omega_{\ell}$ of valuations on the latent variables maps a possible worlds diagram to another possible worlds diagram with the same observed events. The worlds are colored: $\colorline{green}{\wor{0_{\mu}0_{\nu}}}$ green, $\colorline{orange}{\wor{0_{\mu}1_{\nu}}}$ orange, $\colorline{yellow}{\wor{1_{\mu}0_{\nu}}}$ yellow, and $\colorline{violet}{\wor{1_{\mu}1_{\nu}}}$ violet.}
    \label{fig:sym_examples}
\end{figure}

Importantly, the optimizations illuminated above, namely ignoring superfluous specifications and exploiting symmetries, are universal\footnote{As a special case, causal \textit{networks} (which are causal structures where all variables are exogenous or endogenous) contain no superfluous scenarios.}; they can be applied for any causal structure. Additionally, the possible worlds framework intuitively excludes superfluous cases and directly embodies the observational symmetries, making a possible worlds diagram the ideal data structure for performing a search over observed distributions.

\subsection{The Uniformity of Latent Distributions}
\label{sec:uniformity}

The purpose of this section is motivate why it is always possible to approximate any functional causal model $\br{\mathcal G, \mathcal F_{\mathcal V}, \mathcal P_{\mathcal L}}$ with another functional causal model $(\mathcal G, \tilde {\mathcal F}_{\mathcal V}, \tilde{\mathcal P}_{\mathcal L})$ which has latent events $\lambda_{\mathcal L} \in \tilde \Omega_{\mathcal L}$ uniformly distributed. Unsurprisingly, an accurate approximation of this form will require an increase in the cardinality $|\tilde \Omega_{\mathcal L}| > \abs{\Omega_{\mathcal L}}$ of the latent variables.

\begin{definition}[Rational Distributions]
    \label{defn:rational_approximation}
    A discrete probability distribution $\mathtt{P}$ over $\Omega$ is \term{rational} if every probability assigned to events in $\Omega$ by $\mathtt{P}$ is rational,
    \begin{align}
        \forall \lambda \in \Omega, \quad \mathtt{P}(\lambda) = \frac{n_{\lambda}}{d_{\lambda}}, \quad \text{where} \quad n_{\lambda}, d_{\lambda} \in \mathbb{Z}.
    \end{align}
\end{definition}
\begin{definition}[Distance Metric for Distributions]
    Given two probability distributions $\mathtt{P}, \tilde{\mathtt{P}}$ over the same sample space $\Omega$, the distance $\Delta(\mathtt{P}, \tilde{\mathtt{P}})$ between $\mathtt{P}$ and $\tilde{\mathtt{P}}$ is defined as,
    \begin{align}
        \Delta(\mathtt{P}, \tilde{\mathtt{P}}) = \sum_{x \in \Omega} \abs{\mathtt{P}(x) - \tilde{\mathtt{P}}(x)}
    \end{align}
\end{definition}
\begin{theorem}
    \label{thm:discrete_uniform_sampling}
    Let $\mathtt{P}_{\ell} : \Omega_{\ell} \to \bs{0,1}$ be any discrete probability distribution on $\Omega_{\ell}$, then there exists a rational approximation $\tilde {\mathtt{P}}_{\ell} : \Omega_{\ell} \to \bs{0,1}$,
    \begin{align}
        \forall \lambda_{\ell} \in \Omega_{\ell}, \quad \tilde {\mathtt{P}}_{\ell}(\lambda_{\ell}) = \frac{1}{\abs{\Omega_{u}}} \sum_{\omega_{u} \in \Omega_{u}} \delta(\lambda_{\ell}, g(\omega_{u})),
        \label{eq:thm_discrete_uniform_sampling}
    \end{align}
    where $g: \Omega_{u} \to \Omega_{\ell}$ is deterministic and $\Delta(\mathtt{P}_{\ell}, \tilde {\mathtt{P}}_{\ell}) \leq \frac{\abs{\Omega_{u}}-1}{\abs{\Omega_{\ell}}}$.
\end{theorem}
\begin{proof}
    The proof is illustrated in Figure~\ref{fig:discrete_uniform_sampling}. In the special case that $\abs{\Omega_{\ell}} = 1$, the proof is trivial; $g$ simply maps all values of $\omega_{u}$ to the singleton $\lambda_{\ell} \in \Omega_{\ell}$. The proof follows from a construction of $g$ using inverse uniform sampling. Given some ordering $1_{\ell} < 2_{\ell} < \cdots$ of $\Omega_{\ell}$ and ordering $1_{u} < 2_{u} < \cdots$ of $\Omega_{u}$ compute the cumulative distribution function $\mathtt{P}_{\leq \ell}(\lambda_{\ell}) = \sum_{\lambda'_{\ell} \leq \lambda_{\ell} } \mathtt{P}_{\ell}(\lambda'_{\ell})$. Then the function $g : \Omega_{u} \to \Omega_{\ell}$ is defined as,
    \begin{align}
        g(\omega_{u}) = \min\bc{\lambda_{\ell} \in \Omega_{\ell} \mid \mathtt{P}_{\leq \ell}(\lambda_{\ell}) \abs{\Omega_{u}} \geq \omega_{u} }.
    \end{align}
    Consequently, the proportion of $\omega_{u} \in \Omega_{u}$ values which map to $\lambda_{\ell} \in \Omega_{\ell}$ has error $\varepsilon(\lambda_{\ell})$,
    \begin{align}
        \varepsilon(\lambda_{\ell}) = \abs{\Omega_{u}} \mathtt{P}_{\ell}(\lambda_{\ell}) - \abs{g^{-1}(\lambda_{\ell})},
        \label{eq:rational_approximation}
    \end{align}
    where $\abs{\varepsilon(\lambda_{\ell})} \leq 1$ for all $\lambda_{\ell} \in \Omega_{\ell}$ with the exception of the minimum ($1_{\mu}$) and maximum ($\abs{\Omega_{\ell}}_{\ell}$) values where $\abs{\varepsilon(\lambda_{\ell})} \leq 1/2$. Therefore, the proof follows from a direct computation of the distance $\Delta(\mathtt{P}_{\ell}, \tilde {\mathtt{P}}_{\ell})$,
    \begin{align}
        \Delta(\mathtt{P}_{\ell}, \tilde {\mathtt{P}}_{\ell})
        &= \sum_{\lambda_{\ell} \in \Omega_{\ell}} \abs{\mathtt{P}_{\ell}(\lambda_{\ell}) - \tilde {\mathtt{P}}_{\ell}(\lambda_{\ell})}, \\
        &= \sum_{\lambda_{\ell} \in \Omega_{\ell}} \abs{\mathtt{P}_{\ell}(\lambda_{\ell}) - \frac{1}{\abs{\Omega_{u}}} \abs{g^{-1}(\lambda_{\ell})}}, \\
        &= \frac{1}{\abs{\Omega_{u}}} \sum_{\lambda_{\ell} \in \Omega_{\ell}} \abs{\varepsilon(\lambda_{\ell})}, \\
        &\leq \frac{1}{\abs{\Omega_{u}}} \br{\abs{\Omega_{\ell}} - 2 + 2 \frac{1}{2}}, \\
        &= \frac{\abs{\Omega_{\ell}}-1}{\abs{\Omega_{u}}}.
    \end{align}
\end{proof}
\begin{figure}
    \centering
    \includegraphics{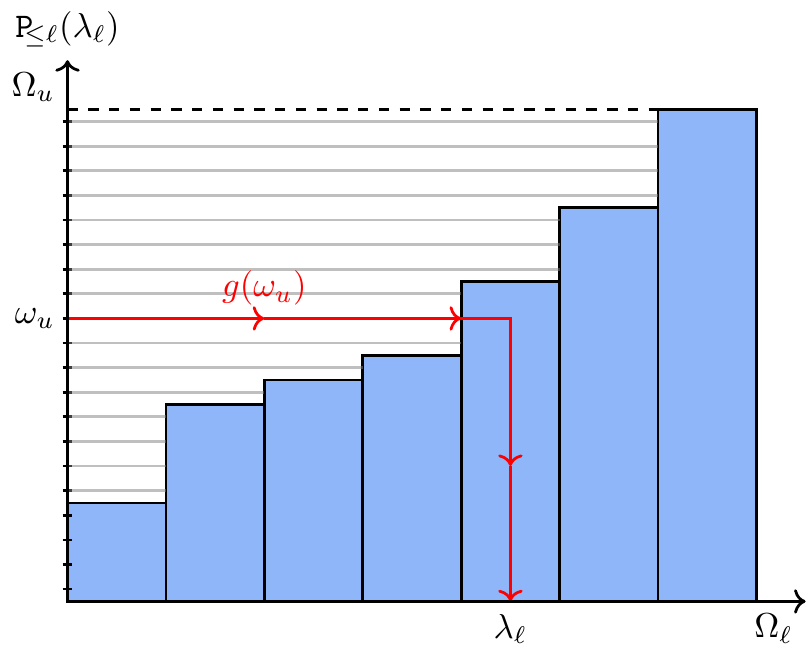}
    \caption{Theorem~\ref{thm:discrete_uniform_sampling}: Approximately sampling a non-uniform distribution using inverse sampling techniques.}
    \label{fig:discrete_uniform_sampling}
\end{figure}

In terms of the causal compatibility problem, Theorem~\ref{thm:discrete_uniform_sampling} suggests that if an observed distribution $\mathtt{P}_{\mathcal V}$ is compatible with $\mathcal G$, and there exists a functional causal model $\br{\mathcal G, \mathcal F_{\mathcal V}, \mathcal P_{\mathcal L}}$ which reproduces $\mathtt{P}_{\mathcal V}$ (via Equation~\ref{eq:functional_compatibility}), then it must be close to a rational distribution $\tilde {\mathtt{P}}_{\mathcal V}$ generated by a functional causal model $(\mathcal G, \tilde {\mathcal F}_{\mathcal V}, \tilde {\mathcal P}_{\mathcal L})$ wherein probability distributions for the latent variables $\tilde {\mathcal P}_{\mathcal L}$ are \textit{uniform}. The following theorem proves this.

\begin{theorem}
    \label{thm:main_error_bound_models}
    Let $\br{\mathcal G, \mathcal F_{\mathcal V}, \mathcal P_{\mathcal L}}$ be a functional causal model with cardinalities $c_{\ell} = \abs{\Omega_{\ell}}$ for the latent variables producing distribution $\mathtt{P}_{\mathcal V}$. Then there exists a functional causal model $(\mathcal G, \tilde {\mathcal F}_{\mathcal V}, \tilde {\mathcal P}_{\mathcal L})$ with cardinalities $k_{\ell} = |\tilde \Omega_{\ell}|$ for the latent variables producing $\tilde {\mathtt{P}}_{\mathcal V}$ where the distributions $\tilde {\mathcal P}_{\mathcal L} = \{\mathsf{U}_{\ell} : \tilde \Omega_{\ell} \to k_{\ell}^{-1} \mid \ell \in \mathcal L\}$ over the latent variables are uniform. In particular, the distance between $\mathtt{P}_{\mathcal V}$ and $\tilde {\mathtt{P}}_{\mathcal V}$ is bounded by,
    \begin{align}
        \Delta(\mathtt{P}_{\mathcal V}, \tilde {\mathtt{P}}_{\mathcal V}) \leq \varepsilon = \sum_{n = 1}^{L} \frac{1}{n!}\br{\frac{L(C - 1)}{K}}^{n} \in \mathcal O\br{\frac{LC}{K}},
        \label{eq:main_error_bound_result}
    \end{align}
    where $C = \max\bc{c_{\ell} \mid \ell \in \mathcal L}$, $K = \min\bc{k_{\ell} \mid \ell \in \mathcal L}$, and $L = \abs{\mathcal L}$ is the number of latent variables.
\end{theorem}

\begin{proof}
    The proof relies on Theorem~\ref{thm:discrete_uniform_sampling} and can be found in Appendix~\ref{sec:proof_of_theorem_main_error_bound}.
\end{proof}

\subsection{A Converging Hierarchy of Compatibility Tests}
\label{sec:a_converging_hierarchy_of_tests}

In Section~\ref{sec:symmetry_and_superfluous}, we discussed how to take advantage of the symmetries of a possible worlds diagram and the superfluities within a set of functional parameters $\mathcal F_{\mathcal V}$ in order to optimally search over functional models. In Section~\ref{sec:uniformity}, we discussed how to approximate any functional causal model $\br{\mathcal G, \mathcal F_{\mathcal V}, \mathcal P_{\mathcal L}}$ using one with uniform latent probability distributions. Here we combine these insights into a hierarchy of probabilistic compatibility tests for the causal compatibility problem.

\begin{definition}
    Given a causal structure $\mathcal G$, and given cardinalities\footnote{The cardinalities for the visible variables, $k_{\mathcal V} = \bc{k_{v} = \abs{\Omega_{v}} \mid v \in \mathcal V}$, are also assumed to be known.} $k_{\mathcal L} = \bc{k_{\ell} = \abs{\Omega_{\ell}} \mid \ell \in \mathcal L}$ for the latent variables, define the \term{uniformly induced distributions}, denoted as $\mathcal U_{\mathcal V}^{(k_{\mathcal L})}(\mathcal G)$, as the set of \textit{all} distributions $\tilde {\mathtt{P}}_{\mathcal V} \in \mathcal M_{\mathcal V}(\mathcal G)$ which admit of a uniform functional model $\br{\mathcal G, \mathcal F_{\mathcal V}, \mathcal P_{\mathcal L}}$ with cardinalities $k_{\mathcal L}$.
\end{definition}

Recall that Section~\ref{sec:symmetry_and_superfluous} demonstrates a method, using the possible worlds framework, for efficient generation of the entirety of $\mathcal U_{\mathcal V}^{(k_{\mathcal L})}(\mathcal G)$.

\begin{lemma}
    \label{lem:epsilon_dense}
    The uniformly induced distributions $\mathcal U_{\mathcal V}^{(k_{\mathcal L})}(\mathcal G)$ form an $\varepsilon$-dense set in $\mathcal M_{\mathcal V}\br{\mathcal G}$,
    \begin{align}
        \mathtt{P}_{\mathcal V} \in \mathcal M_{\mathcal V}\br{\mathcal G} \implies \exists \tilde {\mathtt{P}}_{\mathcal V} \in \mathcal U_{\mathcal V}^{(k_{\mathcal L})}(\mathcal G), \quad \Delta(\mathtt{P}_{\mathcal V}, \tilde {\mathtt{P}}_{\mathcal V}) \leq \varepsilon \in \mathcal O\br{\frac{LC}{K}}
        \label{eq:error_bound_again}
    \end{align}
    where $\varepsilon$ is a function of $K = \min\bc{k_{\ell} \mid \ell \in \mathcal L}$, the number of latent variables $L = \abs{\mathcal L}$, and $C = \max\bc{c_{\ell} \mid \ell \in \mathcal L}$ where $c_{\ell}$ is the minimum upper bound placed on the cardinalities of the latent variable $\ell$ by Theorem~\ref{thm:upper_bound}.
\end{lemma}
\begin{proof}
    Since $c_{\mathcal L} = \bc{c_{\ell} \mid \ell \in \mathcal L}$ are minimum upper bounds placed on the cardinalities of the latent variables by Theorem~\ref{thm:upper_bound}, any $\mathtt{P}_{\mathcal V} \in \mathcal M_{\mathcal V}\br{\mathcal G}$ \textit{must} admit a functional causal model with cardinalities for the latent variables at most $c_{\mathcal L}$. Then by Theorem~\ref{thm:main_error_bound_models}, there exists a uniform causal model producing $\tilde {\mathtt{P}}_{\mathcal V} \in \mathcal U_{\mathcal V}^{(k_{\mathcal L})}(\mathcal G)$, within a distance $\varepsilon$ given by Equation~\ref{eq:main_error_bound_result}.
\end{proof}

Lemma~\ref{lem:epsilon_dense} forms the basis of the following compatibility test,

\begin{theorem}[The Causal Compatibility Test of Order $K$]
    \label{thm:test_K}
    For a probability distribution $\mathtt{P}_{\mathcal V}$ and a causal structure $\mathcal G$, the causal compatibility test of order $K = \min\bc{k_{\ell} \mid \ell \in \mathcal L}$ is defined as the following question:
    \begin{center}
        Does there exist a uniformly induced distribution $\tilde {\mathtt{P}}_{\mathcal V} \in \mathcal U_{\mathcal V}^{(k_{\mathcal L})}(\mathcal G)$ such that $\Delta(\mathtt{P}_{\mathcal V}, \tilde {\mathtt{P}}_{\mathcal V}) \leq \varepsilon\br{K}$?\footnote{Here $\varepsilon\br{K}$ is the value for $\varepsilon$ provided by Lemma~\ref{lem:epsilon_dense}.}
    \end{center}
    As $K \to \infty$, the distance tends to zero $\varepsilon(K) \to 0$ and the sensitivity of the test increases. If $\mathtt{P}_{\mathcal V} \not \in \mathcal M_{\mathcal V}(\mathcal G)$, then $\mathtt{P}_{\mathcal V}$ will fail the test for finite $K$. If $\mathtt{P}_{\mathcal V} \in \mathcal M_{\mathcal V}(\mathcal G)$, then $\mathtt{P}_{\mathcal V}$ will pass the test for all $K$. Moreover, for fixed $K$, the test can readily return the functional causal model behind the best approximation $\tilde {\mathtt{P}}_{\mathcal V}$.
\end{theorem}

First notice that Theorem~\ref{thm:test_K} achieves the same rate of convergence as~\cite{Navascues_2017}. Unlike the result of~\cite{Navascues_2017}, Theorem~\ref{thm:test_K} returns a functional model which approximates $\mathtt{P}_{\mathcal V}$. It is interesting to remark that the distance bound $\varepsilon \in \mathcal O(LC/K)$ in Equation~\ref{eq:error_bound_again} depends on $C = \max\bc{c_{\ell} \mid \ell \in \mathcal L}$ where $c_{\ell}$ is the minimum upper bound placed on the cardinalities of the latent variable $\ell$ by Theorem~\ref{thm:upper_bound}. As conjectured in Appendix~\ref{sec:simplifying_causal_parameters}, it is likely that there are tighter bounds that can be placed on these cardinalities for certain causal structures. Therefore, further research into lowering these bounds will improve the performance of Theorem~\ref{thm:test_K}.

\section{Conclusion}
\label{sec:conclusion}

In conclusion, this paper examined the abstract problem of causal compatibility for causal structures with latent variables. Section~\ref{sec:possible_worlds_framework} introduced the framework of possible worlds in an effort to provide solutions to the causal compatibility problem. Central to this framework is the notion of a possible worlds diagram, which can be viewed as a hybrid between a causal structure and the functional parameters of a causal model. It does not however, convey any information about the probability distributions over the latent variables.

In Section~\ref{sec:complete_possibilistic}, we utilized the possible worlds framework to prove possibilistic incompatibility of a number of examples. In addition, we demonstrated the utility of our approach by resolving an open problem associated with one of Evans'~\cite{Evans_2016} causal structures. Particularly, we have shown the causal structure in Figure~\ref{fig:evans_causal_structure_1} is incompatible with the distribution in Equation~\ref{eq:evans_failing_distribution}. Section~\ref{sec:complete_possibilistic} concluded with an algorithm for completely solving the possibilistic causal compatibility problem.

In Section~\ref{sec:complete_probabilistic}, we discussed how to efficiently search through the observational equivalence classes of functional parameters using a possible worlds diagram as a data structure. Afterwards, we derived bounds on the distance between compatible distributions and uniformly induced ones. By combining these results, we provide a hierarchy of necessary tests for probabilistic causal compatibility which converge in the limit.

\section{Acknowledgments}
Foremost, I must thank my supervisor Robert W. Spekkens for his unwavering support and encouragement. Second, I would like to sincerely thank Elie Wolfe for our numerous and lengthy discussions. Without him or his research, this paper simply would not exist. Finally, I thank the two anonymous referees for providing insight necessary for significantly improving this paper.

\printbibliography
\clearpage
\appendix

\section{Simplifying Causal Structures}
\label{sec:simplifying_causal_structures}

\subsection{Observational Equivalence}
\label{sec:observational_impact}
From an experimental perspective, a causal model $\br{\mathcal G, \mathcal P}$ has the ability to predict the effects of \textit{interventions}; by manually tinkering with the configuration of a system, one can learn more about the underlying mechanisms than from observations alone~\cite{Pearl_2009}. When interventions become impossible, because experimentation is expensive or unethical for example, it becomes possible for distinct causal structures to admit the same set of compatible correlations. An important topic in the study of causal inference is the identification of \textit{observationally equivalent} causal structures. Two causal structures $\mathcal G$ and $\mathcal G'$ are observationally equivalent or simply \textit{equivalent} if they share the same set of compatible models $\mathcal M_{\mathcal V}\br{\mathcal G} = \mathcal M_{\mathcal V}\br{\mathcal G'}$. For example, the direct cause causal structure in Figure~\ref{fig:direct_cause} is observationally equivalent to the common cause causal structure in Figure~\ref{fig:common_cause}. Identifying observationally equivalent causal structures is of fundamental importance to the causal compatibility problem; if a distribution $\mathtt{P}_{\mathcal V}$ is known to satisfy the hypotheses of $\mathcal G$, and $\mathcal M_{\mathcal V}(\mathcal G) = \mathcal M_{\mathcal V}(\mathcal G')$ then it will also satisfy the hypotheses of $\mathcal G'$.

\begin{figure}
    \centering
    \begin{subfigure}[t]{0.49\textwidth}
        \centering
        \includegraphics{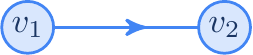}
        \caption{A direct cause from $v_1$ to $v_2$.}
        \label{fig:direct_cause}
    \end{subfigure}
    \begin{subfigure}[t]{0.49\textwidth}
        \centering
        \includegraphics{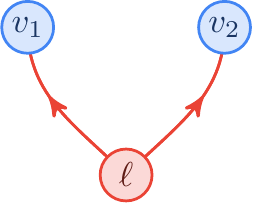}
        \caption{A shared common cause $\ell$ between $v_1$ and $v_2$.}
        \label{fig:common_cause}
    \end{subfigure}
    \caption{The causal structures of (a) and (b) are observationally equivalent.}
    \label{fig:direct_cause_common_cause}
\end{figure}

\subsection{Exo-Simplicial Causal Structures}

In general, other than being a directed acyclic graph, there are no restrictions placed on a causal structure with latent variables. Nonetheless,~\cite{Evans_2016} demonstrated a number of transformations on causal structures which leave $\mathcal M_{\mathcal V}\br{\mathcal G}$ invariant. Two of these transformations are the subject of interest for this section. The first concerns itself with latent vertices that have parents while the second concerns itself with parent-less latent vertices that share children. Each will be taken in turn.

\begin{figure}
    \centering

    \begin{subfigure}[t]{0.49\textwidth}
        \centering
        \includegraphics{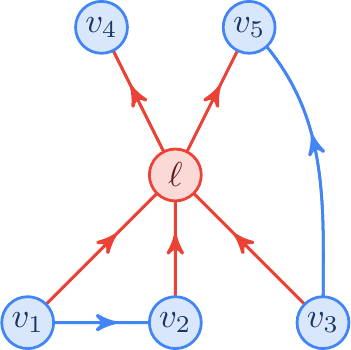}
        \caption{A latent vertex with observable parents.}
        \label{fig:lat_with_obs_parents}
    \end{subfigure}
    \begin{subfigure}[t]{0.49\textwidth}
        \centering
        \includegraphics{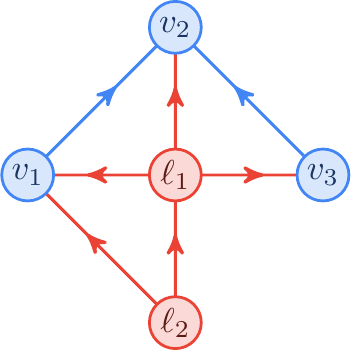}
        \caption{A latent vertex with latent parents}
        \label{fig:latent_edge}
    \end{subfigure}
    \begin{subfigure}[t]{0.49\textwidth}
        \centering
        \includegraphics{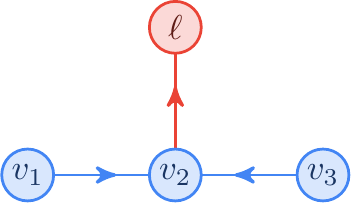}
        \caption{A latent vertex with no children.}
        \label{fig:latent_no_children}
    \end{subfigure}
    \begin{subfigure}[t]{0.49\textwidth}
        \centering
        \includegraphics{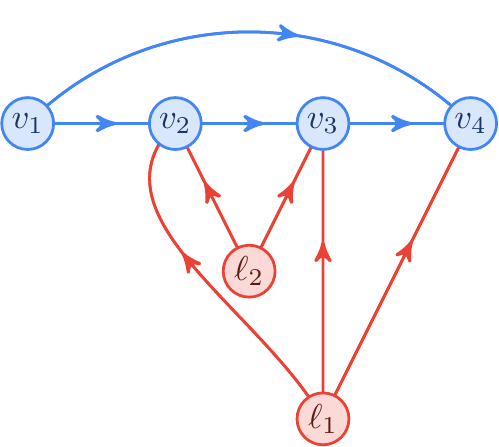}
        \caption{Latent vertices with nested children.}
        \label{fig:latent_containment}
    \end{subfigure}
    \caption{Examples of causal structures which are not exo-simplicial.}
    \label{fig:non_simplicial}
\end{figure}

\begin{definition}[See Defn. 3.6~\cite{Evans_2016}]
    Given a causal structure $\mathcal G = \br{\mathcal V \cup \mathcal L, \mathcal E}$ with latent vertex $\ell \in \mathcal L$, the \term{exogenized} causal structure $\exo[\mathcal G]{\ell}$ is formed by taking $\mathcal E$ and (i) adding an edge $p \to c$ for every $p \in \pa[\mathcal G]{\ell}$ and $c \in \ch[\mathcal G]{\ell}$ if not already present, and (ii) deleting all edges of the form $p \to \ell$ where $p \in \pa[\mathcal G]{\ell}$. If $\pa[\mathcal G]{\ell}$ is empty, $\exo[\mathcal G]{\ell} = \mathcal G$.
\end{definition}

\begin{lemma}[See Lem. 3.7~\cite{Evans_2016}]
    Given a causal structure $\mathcal G = \br{\mathcal V \cup \mathcal L, \mathcal E}$ with latent vertex $\ell \in \mathcal L$, then $\mathcal M_{\mathcal V}\br{\exo[\mathcal G]{\ell}} = \mathcal M_{\mathcal V}\br{\mathcal G}$.
    \label{lem:exogenization}
\end{lemma}

\begin{proof}
    See proof of Lem. 3.7 from~\cite{Evans_2016}.
\end{proof}

The concept of exogenization is best understood with an example.
\begin{example}
    Consider the causal structure $\mathcal G_{\ref{fig:lat_with_obs_parents}}$ in Figure~\ref{fig:lat_with_obs_parents}. In $\mathcal G_{\ref{fig:lat_with_obs_parents}}$, the latent variable $\ell$ has parents $\pa{\ell} = \bc{v_{1}, v_{2}, v_{3}}$ and children $\ch{\ell} = \bc{v_{4}, v_{5}}$. Since the sample space $\Omega_{\ell}$ is unknown, its cardinality could be arbitrarily large or infinite. As a result, it has an unbounded capacity to \textit{inform} its children of the valuations of its parents, e.g. $v_{4}$ can have complete knowledge of $v_{1}$ through $\ell$ and therefore adding the edge $v_{1} \to v_{4}$ has no observational impact. Applying similar reasoning to all parents of $\ell$, i.e. applying Lemma~\ref{lem:exogenization}, one converts $\mathcal G_{\ref{fig:lat_with_obs_parents}}$ to the observationally equivalent, exogenized causal structure $\exo[\mathcal G_{\ref{fig:lat_with_obs_parents}}]{\ell}$ depicted in Figure~\ref{fig:exogenized_example}.
\end{example}
\begin{figure}
    \centering
    \includegraphics{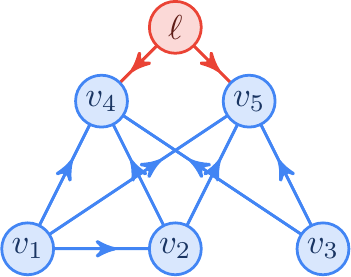}
    \caption{The exogenized causal structure $\exo[\mathcal G_{\ref{fig:lat_with_obs_parents}}]{\ell}$.}
    \label{fig:exogenized_example}
\end{figure}
Lemma~\ref{lem:exogenization} can be applied recursively to each latent variable $\ell \in \mathcal L$ in order to transform any causal structure $\mathcal G$ into an observationally equivalent one wherein the latent variables have no parents (exogenous). Notice that the process of exogenization also works when latent vertices have latent parents, as is the case in Figure~\ref{fig:latent_edge}. Also, when a latent vertex $\ell$ has no children, the process of exogenization disconnects $\ell$ from the rest of the causal structure, where it can be ignored with no observational impact due to Equation~\ref{eq:latent_marginalization_continuous}.

The next observationally invariant transformation requires the exogenization procedure to have been applied first. In Figure~\ref{fig:latent_containment}, $\ell_{1}$ and $\ell_{2}$ are exogenous latent variables where $\ch[\mathcal G_{\ref{fig:latent_containment}}]{\ell_{2}} \subset \ch[\mathcal G_{\ref{fig:latent_containment}}]{\ell_{1}}$. Therefore, because the sample space $\Omega_{\ell_{1}}$ is unspecified, it has the capacity to emulate any dependence that $v_{3}$ and/or $v_{2}$ might have on $\ell_{2}$. This idea is captured by Lemma~\ref{lem:simplicialization}.

\begin{lemma}[See Lem. 3.8~\cite{Evans_2016}]
    Let $\mathcal G$ be a causal structure with latent vertices $\ell, \ell' \in \mathcal L$ where $\ell \neq \ell'$. If $\pa[\mathcal G]{\ell} = \pa[\mathcal G]{\ell'} = \emptyset$, and $\ch[\mathcal G]{\ell'} \subseteq \ch[\mathcal G]{\ell}$ then $\mathcal M_{\mathcal V}\br{\mathcal G} = \mathcal M_{\mathcal V}\br{\sub[\mathcal G]{\mathcal V \cup \mathcal L - \bc{\ell'}}}$.
    \label{lem:simplicialization}
\end{lemma}

\begin{proof}
    See proof of Lem. 3.8 from~\cite{Evans_2016}.
\end{proof}

An immediate corollary of Lemma~\ref{lem:simplicialization} is that the latent variables $\bc{\ell \mid \ell \in \mathcal L}$, which are isomorphic to their children $\bc{\ch{\ell} \mid \ell \in \mathcal L}$, are isomorphic to the facets of a simplicial complex over the visible variables.

\begin{definition}
    An \term{(abstract) simplicial complex}, $\Delta$, over a finite set $\mathcal V$ is a collection of non-empty subsets of $\mathcal V$ such that:
    \begin{enumerate}
        \item $\{v\} \in \Delta$ for all $v \in \mathcal V$; and
        \item if $C_1 \subseteq C_2 \subseteq \mathcal V$, $C_2 \in \Delta \Rightarrow C_1 \in \Delta$.
    \end{enumerate}
    The maximal subsets with respect to inclusion are called the \term{facets} of the simplicial complex.
\end{definition}

In~\cite{Evans_2016}, this concept led to the invention of \textit{mDAGs} (or marginal directed acyclic graphs), a hybrid between a directed acyclic graph and a simplicial complex. In this work, we refrain from adopting the formalism of mDAGs and instead continue to consider causal structures as entirely directed acyclic graphs. Despite this refrain, Lemmas~\ref{lem:exogenization},~\ref{lem:simplicialization} demonstrate that for the purposes of the causal compatibility problem, the latent variables of a causal structure can be assumed to be exogenous and to have children forming the facets of a simplicial complex. Causal structures which adhere to this characterization will be referred to as \textit{exo-simplicial} causal structures. Figure~\ref{fig:simplicial} depicts four exo-simplicial causal structures respectively equivalent to the causal structures in Figure~\ref{fig:non_simplicial}.

\begin{figure}
    \centering
    \begin{subfigure}[t]{0.49\textwidth}
        \centering
        \includegraphics{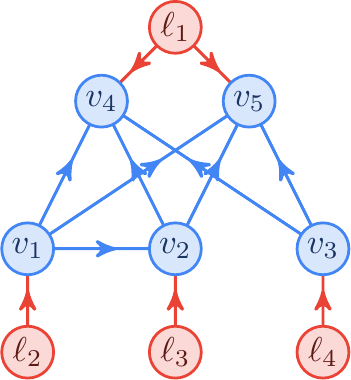}
        \caption{}
    \end{subfigure}
    \begin{subfigure}[t]{0.49\textwidth}
        \centering
        \includegraphics{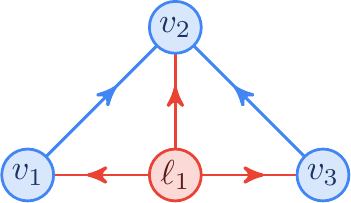}
        \caption{}
    \end{subfigure}
    \begin{subfigure}[t]{0.49\textwidth}
        \centering
        \includegraphics{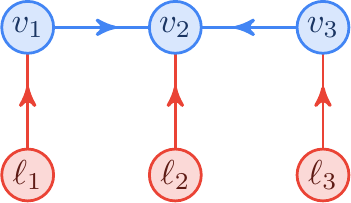}
        \caption{}
    \end{subfigure}
    \begin{subfigure}[t]{0.49\textwidth}
        \centering
        \includegraphics{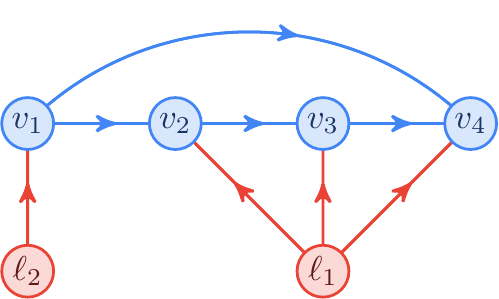}
        \caption{}
    \end{subfigure}
    \caption{Examples of exo-simplicial causal structures which are observationally equivalent to their respective counterparts in Figure~\ref{fig:non_simplicial}.}
    \label{fig:simplicial}
\end{figure}

\section{Simplifying Causal Parameters}
\label{sec:simplifying_causal_parameters}

Recall that a causal model $(\mathcal G, \mathcal P)$ consists of a causal structure $\mathcal G$ and causal parameters $\mathcal P$. Appendix~\ref{sec:simplifying_causal_structures} simplified the causal compatibility problem by revealing that each causal structure $\mathcal G$ can be replaced with an observationally equivalent exo-simplicial causal structure $\mathcal G'$ such that $\mathcal M_{\mathcal V}(\mathcal G) = \mathcal M_{\mathcal V}(\mathcal G')$. The purpose of this section is to simplify the causal compatibility problem in three ways. Section~\ref{sec:determinism} demonstrates that the visible causal parameters $\bc{\mathtt{P}_{v|\pa{v}} \mid v \in \mathcal V}$ of a causal model can be assumed to be deterministic without observational impact. Section~\ref{sec:finite_bounds} shows that if the observed distribution is finite (i.e. $\abs{\Omega_{\mathcal V}} < \infty$), one only needs to consider \textit{finite} probability distributions for the latent variables. Moreover, explicit upper bounds on the cardinalities of the latent variables can be computed.

\subsection{Determinism}
\label{sec:determinism}

\begin{lemma}
    \label{lem:determinism}
    If $\mathtt{P}_{\mathcal V} \in \mathcal M_{\mathcal V}\br{\mathcal G}$ and $\mathcal G$ is exo-simplicial (see Appendix~\ref{sec:simplifying_causal_structures}), then without loss of generality, the causal parameters $\mathtt{P}_{v|\pa[\mathcal G]{v}}$ over the observed variables can be assumed to be deterministic, and consequently,
    \begin{align}
        \forall x_{\mathcal V} \in \Omega_{\mathcal V}, \quad \mathtt{P}_{\mathcal V}(x_{\mathcal V}) = \prod_{\ell \in \mathcal L} \int_{\lambda_{\ell} \in \Omega_{\ell}} \mathrm{d}\mathtt{P}_{\ell}(\lambda_{\ell}) \prod_{v \in \mathcal L} \delta(x_{v}, f_{v}(x_{\vpa[\mathcal G]{v}},\lambda_{\lpa[\mathcal G]{v}}))
    \end{align}
\end{lemma}
\begin{proof}
    Since $\mathtt{P}_{\mathcal V} \in \mathcal M_{\mathcal V}\br{\mathcal G}$, by definition, there exists a joint distribution $\mathtt{P}_{\mathcal V \cup \mathcal L}$ (or density $\mathrm{d} \mathtt{P}_{\mathcal V \cup \mathcal L}$) admitting marginal $\mathtt{P}_{\mathcal V}$ via Equation~\ref{eq:latent_marginalization_continuous}. Since the joint distribution satisfies Equation~\ref{eq:causal_parameter_product}, it is possible to associate to each observed variable $\mathsf{X}_{v}$ an independent random variable $\mathsf{E}_{e_{v}}$ and measurable function $f_{v} : \Omega_{\vpa[\mathcal G]{v}} \times \Omega_{\lpa[\mathcal G]{v}} \times \Omega_{e_{v}}$ such that for all $v \in \mathcal V$,
    \begin{align}
        \mathsf{X}_{v} = f_{v}\br{\mathsf{X}_{\vpa[\mathcal G]{v}}, \mathsf{\Lambda}_{\lpa[\mathcal G]{v}}, \mathsf{E}_{e_v}}.
        \label{eq:functional_model}
    \end{align}
    Therefore, by promoting each $e_v$ to the status of a latent variable in $\mathcal G$ and adding an edge $e_{v} \to v$ to $\mathcal E$, each $\mathsf{X}_{v}$ becomes a deterministic function of its parents. Finally, making use of the fact that $\mathcal G$ is exo-simplicial, every error variable $e_{v}$ has its children $\ch[\mathcal G]{e_{v}} = \bc{v}$ nested inside the children of at least one other pre-existing latent variable. Therefore, by applying Lemma~\ref{lem:simplicialization}, $e_{v}$ is eliminated and one recovers the original $\mathcal G$.
\end{proof}

Essentially, Lemma~\ref{lem:determinism} indicates that any non-determinism due to local noise variables $\mathsf{E}_{e_v}$ can be emulated by the behavior of the latent variables $\mathcal L$.

\subsection{The Finite Bound for Latent Cardinalities}
\label{sec:finite_bounds}

In~\cite{Rosset_2017}, it was shown that if the visible variables have finite cardinality (i.e. $k_{\mathcal V} = \abs{\Omega_{\mathcal V}}$ is finite), then for a particular class of causal structures known as \term{causal networks}, the cardinalities of the latent variables could be assumed to be finite as well. A causal network is a causal structure where all latent variables have no parents (are exogenous) and all visible variables either have no parents or no children~\cite{Navascues_2017}. The purpose of this section is to generalize the results of~\cite{Rosset_2017} to the case of exo-simplicial causal \textit{structures}. Although the proof techniques presented here are similar to that of~\cite{Rosset_2017}, the best upper bounds placed on $k_{\mathcal L} = \abs{\Omega_{\mathcal L}}$ depends more intimately on the form of $\mathcal G$. It is also anticipated that the upper bounds presented here are sub-optimal, much like~\cite{Rosset_2017}. It is also worth noting that the results presented here hold independently of whether or not Lemma~\ref{lem:determinism} is applied.

\begin{theorem}
    \label{thm:upper_bound}
    Let $\br{\mathcal G, \mathcal P}$ be a causal model with (possibly infinite) cardinalities $k_{\mathcal L} = \bc{k_{\ell} \mid \ell \in \mathcal L}$ for the latent variables such that,
    \begin{align}
        \forall x_{\mathcal V} \in \Omega_{\mathcal V}, \quad \mathtt{P}_{\mathcal V}(x_{\mathcal V}) = \prod_{\ell \in \mathcal L} \int_{\lambda_{\ell} \in \Omega_{\ell}} \mathrm{d} \mathtt{P}_{\ell}(\lambda_{\ell}) \prod_{v \in \mathcal V} \mathtt{P}_{v|\pa{v}}(x_{v} | x_{\vpa{v}}\lambda_{\lpa{v}}), \label{eq:main_compatibility_test}
    \end{align}
    produces the distribution $\mathtt{P}_{\mathcal V}$. Then there exists a causal model $\br{\mathcal G, \mathcal P'}$ reproducing $\mathtt{P}_{\mathcal V}$ with cardinalities $k_{\mathcal L} = \bc{k_{\ell} \mid \ell \in \mathcal L}$ where each $k_{\ell}$ is a finite.
\end{theorem}
\begin{figure}
    \centering
    \includegraphics{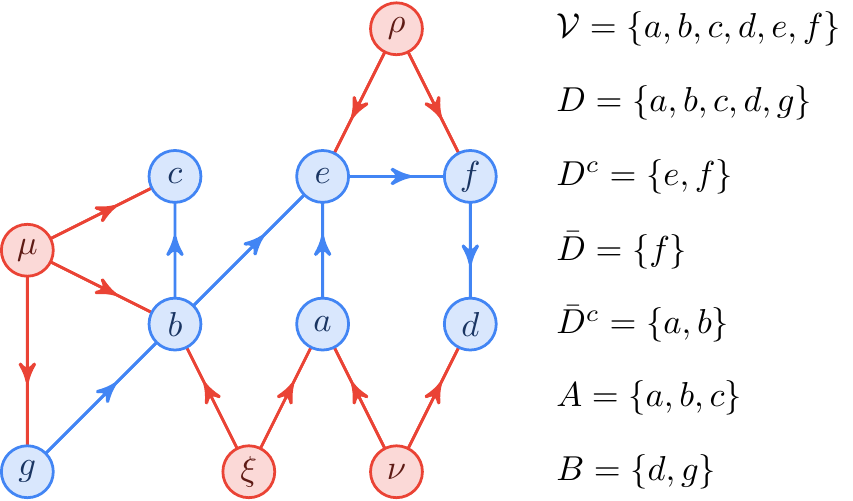}
    \caption{A causal structure $\mathcal G_{\ref{fig:upper_bound_concept}}$ that helps in visualizing the proof of Theorem~\ref{thm:upper_bound}.}
    \label{fig:upper_bound_concept}
\end{figure}
\begin{proof}
    The following proof considers each latent variable $\xi \in \mathcal L$ independently and obtains a value for $k_{\ell}$ in each case. Let $\mathcal L' = \mathcal L - \bc{\xi}$ denote the set of latent variables with $\xi$ removed. Let $\mathrm{d} \mathtt{P}_{\mathcal L'} = \prod_{\ell \in \mathcal L'} \mathrm{d} \mathtt{P}_{\ell}$ be a probability density over $\Omega_{\mathcal L'}$ and consider the conditional probability distribution $\mathtt{P}_{\mathcal V | \xi}(x_{\mathcal V}|\lambda_{\xi})$ given $\lambda_{\xi}$,
    \begin{align}
        \mathtt{P}_{\mathcal V | \xi}(x_{\mathcal V}|\lambda_{\xi})
        &= \int_{\Omega_{\mathcal L'}} \mathrm{d} \mathtt{P}_{\mathcal L'}(\lambda_{\mathcal L'}) \prod_{v \in \mathcal V} \mathtt{P}_{v|\pa{v}}(x_{v} | x_{\vpa{v}}\lambda_{\lpa{v}}) \label{eq:card_1}
    \end{align}
    Consulting Figure~\ref{fig:upper_bound_concept} for clarity, define the \textit{district} $D \subseteq \mathcal V$ of $\xi$ to be the maximal set of visible vertices $v$ in $\mathcal G$ for which there exists an \textit{undirected} path from $v$ to $\xi$ with alternating visible/latent vertices. Let $D^{c} = \mathcal V - D$, $\bar D = \pa{D} - D$ and $\bar D^{c} = \pa{D^{c}} - D^{c}$. The district $D$ has the property that $\mathtt{P}_{\mathcal V | \xi}$ factorizes over $D, D^{c}$~\cite{Evans_2016},
    \begin{align}
        \mathtt{P}_{\mathcal V | \xi}(x_{\mathcal V}|\lambda_{\xi}) = \mathtt{P}_{D | \bar D \xi}(x_{D}|x_{\bar D}\lambda_{\xi})\mathtt{P}_{D^{c}|\bar D^{c}}(x_{D^{c}}|x_{\bar{D}^{c}}). \label{eq:card_2}
    \end{align}
    For varying $\lambda_{\xi}$, consider a vector representation $p_{\lambda_{\xi}}$ of the conditional distribution $\mathtt{P}_{D | \bar D \xi}\br{x_{D} | x_{\bar D} \lambda_{\xi}}$ and define $U = \bc{p_{\lambda_{\xi}} \mid \lambda_{\xi} \in \Omega_{\xi}}$. By construction, the center of mass $p^{*}$ of $U$ represents $\mathtt{P}_{D|\bar D}(x_{D} | x_{\bar D})$,
    \begin{align}
        p^{*} &= \int_{\Omega_{\xi}} \mathrm{d} \mathtt{P}_{\xi}(\lambda_{\xi}) p_{\lambda_{\xi}} \\
        \mathtt{P}_{D|\bar D}(x_{D} | x_{\bar D}) &= \int_{\Omega_{\xi}} \mathrm{d} \mathtt{P}_{\xi}(\lambda_{\xi}) \mathtt{P}_{D | \bar D \xi}(x_{D} | x_{\bar D}\lambda_{\xi}) \label{eq:card_3}
    \end{align}
    Therefore, by a variant of Carathéodory's theorem due to Fenchel~\cite{Barany_2012}, if $U$ is compact and connected, then $p^{*}$ can be written as a finite convex decomposition,
    \begin{align}
        p^{*} = \sum_{j=1}^{\aff{U}} w_{j} p_{j}, \quad \sum_{j} w_{j} = 1, \quad \forall i, w_{i} \geq 0.
    \end{align}
    where $\aff{U}$ is the affine dimension of $U$. Then by letting $\Omega_{\xi} = \{0_{\xi}, 1_{\xi}, \ldots, \aff{U}_{\xi}\}$ be a finite sample space for $\xi$ distributed according to $\mathtt{P}_{\xi}(\lambda_{\xi}) = w_{\lambda}$, by Equations~\ref{eq:main_compatibility_test},~\ref{eq:card_1},~\ref{eq:card_2} and~\ref{eq:card_3},
    \begin{align}
        \mathtt{P}_{\mathcal V}(x_{\mathcal V}) &= \sum_{\lambda_{\xi} \in \Omega_{\xi}} \mathtt{P}_{\xi}(\lambda_{\xi}) \mathtt{P}_{\mathcal V | \xi}(x_{\mathcal V}|\lambda_{\xi}).
    \end{align}
    Therefore, causal parameters exist reproducing $\mathtt{P}_{\mathcal V}$ with cardinality $k_{\xi} = \aff{U}$. What remains is to show that $U$ is compact and to find a bound on $\aff{U}$.

    Because of normalization constraints on each $p_{\lambda_{\xi}}$, $U$ is bounded. Moreover,~\cite{Rosset_2017} demonstrates that $U$ can be taken to be closed as well. Again consulting Figure~\ref{fig:upper_bound_concept} for clarity, partition $D$ into subsets $A = \des{\xi} \cap D$ and $B = D - A$. This partitioning enables one to identify the following linear equality constraint placed on all points $p_{\lambda_\xi}$:
    \begin{align}
        &\sum_{x_{A} \in \Omega_{A}} \mathtt{P}_{D | \bar D \xi}(x_{D}|x_{\bar D}\lambda_{\xi}) \\
        &\quad= \sum_{x_{A} \in \Omega_{A}} \mathtt{P}_{A | B \bar D \xi}(x_{A}|x_{B}x_{\bar D}\lambda_{\xi})\mathtt{P}_{B | \bar D \xi}(x_{B}|x_{\bar D} \lambda_{\xi}) \\
        &\quad= \mathtt{P}_{B | \bar D \xi}(x_{B}|x_{\bar D} \lambda_{\xi}) \\
        &\quad= \mathtt{P}_{B | \bar D}(x_{B}|x_{\bar D}), \label{eq:marginal_equality_constraint_proof}
    \end{align}
    where the last equality holds because $B$ is independent of $\xi$ given $\bar D$\footnote{Every path from $b \in B$ to $\xi$ must pass through an unconditioned collider in $A$ and therefore the d-separation relation $B \perp \bc{\xi} \mid \bar D$ holds~\cite{Pearl_2009}.}.
    Furthermore note that if $U$ is not connected, it can be made connected by a scheme due to~\cite{Rosset_2017} which adds noisy variants of each $p_{\lambda_{\xi}}$ to $U$. Simply include a noise parameter $\nu \in \bs{0,1}$ such that $\lambda_\xi' = \br{\lambda_\xi, \nu}$ and adjust the response functions for variables in $A$ such that,
    \begin{align}
        \mathtt{P}_{A | B \bar D \xi}(x_{A}|x_{B}x_{\bar D}\lambda_{\xi}\nu) = \nu \mathtt{P}_{A | B \bar D \xi}(x_{A}|x_{B}x_{\bar D}\lambda_{\xi}) + \frac{1- \nu}{\abs{\Omega_{A}}}
        \label{eq:noise_model}
    \end{align}
    For each degree of noise $0 \leq \nu \leq 1$, Equation~\ref{eq:noise_model} defines a noisy model $p_{\lambda_{\xi},\nu}$ which are added to $U$. As special cases, no noise $\nu = 0$, yields $p_{\lambda_{\xi},0} = p_{\lambda_{\xi}} \in U$ and complete noise $\nu = 1$ yields $p_{\lambda_{\xi},1}$ representing $\mathtt{P}_{B | \bar D}(x_{B}|x_{\bar D})/\abs{\Omega_{A}} \in U$ which is \textit{independent} of $\lambda_{\xi}$. Therefore, $U$ is connected. Finally, the affine dimension $\aff{U}$ is at most the affine dimension of $\mathtt{P}_{D | \bar D}$ with the degrees of freedom associated with satisfying Equation~\ref{eq:marginal_equality_constraint_proof} removed~\cite{Rosset_2017}. Therefore,
    \begin{align}
        k_{\xi} = \aff{U} \leq \aff{\mathtt{P}_{D | \bar D}} - \aff{\mathtt{P}_{B | \bar D}}
    \end{align}
\end{proof}

\section{Proof of Theorem~\ref{thm:main_error_bound_models}}
\label{sec:proof_of_theorem_main_error_bound}

\begin{proof}
    The proof first constructs the distribution $\tilde {\mathtt{P}}_{\mathcal V}$ which satisfies the error bound in Equation~\ref{eq:main_error_bound_result}. Afterwards, a uniform functional model $(\mathcal G, \tilde {\mathcal F}_{\mathcal V}, \tilde {\mathcal P}_{\mathcal L})$ is constructed which produces $\tilde {\mathtt{P}}_{\mathcal V}$. Begin by letting $\tilde {\mathtt{P}}_{\ell}$ denote the rational approximation of $\mathtt{P}_{\ell}$ for each $\ell \in \mathcal L$ as prescribed by Theorem~\ref{thm:discrete_uniform_sampling}. Then, let
    \begin{align}
        \mathtt{P}_{\mathcal L}(\lambda_{\mathcal L}) = \prod_{\ell \in \mathcal L} \mathtt{P}_{\ell}(\lambda_{\ell}), \quad
        \tilde {\mathtt{P}}_{\mathcal L}(\lambda_{\mathcal L}) = \prod_{\ell \in \mathcal L} \tilde {\mathtt{P}}_{\ell}(\lambda_{\ell}).
    \end{align}
    The joint distribution $\mathtt{P}_{\mathcal V}$ and the rational approximation $\tilde {\mathtt{P}}_{\mathcal V}$ are then given by,
    \begin{align}
        \mathtt{P}_{\mathcal V}(x_{\mathcal V}) &= \sum_{\lambda_{\mathcal L} \in \Omega_{\mathcal L}} \mathtt{P}_{\mathcal L}(\lambda_{\mathcal L}) \delta(x_{\mathcal V}, \mathcal F_{\mathcal V}(\lambda_{\mathcal L}) ), \\
        \tilde {\mathtt{P}}_{\mathcal V}(x_{\mathcal V}) &= \sum_{\lambda_{\mathcal L} \in \Omega_{\mathcal L}} \tilde {\mathtt{P}}_{\mathcal L}(\lambda_{\mathcal L}) \delta(x_{\mathcal V}, \mathcal F_{\mathcal V}(\lambda_{\mathcal L}) ). \label{eq:the_tilde_that_needs_to_match}
    \end{align}
    The distance $\Delta(\mathtt{P}_{\mathcal V}, \tilde {\mathtt{P}}_{\mathcal V})$ between the visible joint distributions is no greater than the distance $\Delta(\mathtt{P}_{\mathcal L}, \tilde {\mathtt{P}}_{\mathcal L})$ between the latent joint distributions:
    \begin{align}
        \Delta(\mathtt{P}_{\mathcal V}, \tilde {\mathtt{P}}_{\mathcal V})
        &= \sum_{x_{\mathcal V} \in \Omega_{\mathcal V}} \abs{\mathtt{P}_{\mathcal V}(x_{\mathcal V}) - \tilde {\mathtt{P}}_{\mathcal V}(x_{\mathcal V})} \\
        &= \sum_{x_{\mathcal V} \in \Omega_{\mathcal V}} \abs{\sum_{\lambda_{\mathcal L} \in \Omega_{\mathcal L}} \bc{\mathtt{P}_{\mathcal L}(\lambda_{\mathcal L}) - \tilde {\mathtt{P}}_{\mathcal L}(\lambda_{\mathcal L})} \delta(x_{\mathcal V}, \mathcal F_{\mathcal V}(\lambda_{\mathcal L}) )} \\
        &\leq \sum_{\lambda_{\mathcal L} \in \Omega_{\mathcal L}}\sum_{x_{\mathcal V} \in \Omega_{\mathcal V}} \abs{ \mathtt{P}_{\mathcal L}(\lambda_{\mathcal L}) - \tilde {\mathtt{P}}_{\mathcal L}(\lambda_{\mathcal L}) }\delta(x_{\mathcal V}, \mathcal F_{\mathcal V}(\lambda_{\mathcal L}) ) \\
        &= \sum_{\lambda_{\mathcal L} \in \Omega_{\mathcal L}} \abs{ \mathtt{P}_{\mathcal L}(\lambda_{\mathcal L}) - \tilde {\mathtt{P}}_{\mathcal L}(\lambda_{\mathcal L}) } \\
        &= \Delta(\mathtt{P}_{\mathcal L}, \tilde {\mathtt{P}}_{\mathcal L})
        \label{eq:error_is_only_latent}
    \end{align}
    The bound in Equation~\ref{eq:main_error_bound_result} will be derived using Equation~\ref{eq:rational_approximation}. For convenience of notation, let the latent variables be indexed $\mathcal L = \bc{\ell_{1}, \ell_{2}, \ldots, \ell_{L}}$ and let $\mathcal L' = \bc{u_{1}, u_{2}, \ldots, u_{L}}$ index the corresponding uniformly distributed variables as defined in Theorem~\ref{thm:discrete_uniform_sampling}. Then,
    \begin{align}
        &\Delta(\mathtt{P}_{\mathcal L}, \tilde {\mathtt{P}}_{\mathcal L}) \\
        &\quad= \sum_{\lambda_{\mathcal L} \in \Omega_{\mathcal L}} \abs{\mathtt{P}_{\mathcal L}(\lambda_{\mathcal L}) - \tilde {\mathtt{P}}_{\mathcal L}(\lambda_{\mathcal L})} \\
        &\quad= \sum_{\lambda_{\mathcal L} \in \Omega_{\mathcal L}} \abs{\prod_{j=1}^{L} \mathtt{P}_{\ell_{j}}(\lambda_{\ell_{j}}) - \prod_{j=1}^{L} \tilde {\mathtt{P}}_{\ell_{j}}(\lambda_{\ell_{j}})} \\
        &\quad= \sum_{\lambda_{\mathcal L} \in \Omega_{\mathcal L}} \abs{\prod_{j=1}^{L} \br{\tilde {\mathtt{P}}_{\ell_{j}}(\lambda_{\ell_{j}}) + \frac{\varepsilon(\lambda_{\ell_{j}})}{\abs{\Omega_{u_{j}}}}} - \prod_{j=1}^{L} \tilde {\mathtt{P}}_{\ell_{j}}(\lambda_{\ell_{j}})}
        \label{eq:left_off}
    \end{align}
    Here it becomes advantageous to define helper variables $\Gamma_{0, j}$ and $\Gamma_{1, j}$ such that,
    \begin{align}
        \Gamma_{0, j}(\lambda_{\mathcal L}) = \tilde {\mathtt{P}}_{\ell_{j}}(\lambda_{\ell_{j}}), \quad \Gamma_{1, j}(\lambda_{\mathcal L}) = \frac{\varepsilon(\lambda_{\ell_{j}})}{\abs{\Omega_{u_{j}}}}. \label{eq:gamma_values}
    \end{align}
    Additionally, let $b \in \bc{0,1}^{L}$ be a binary string of length $L$. Then Equation~\ref{eq:left_off} becomes,
    \begin{align}
        &\Delta(\mathtt{P}_{\mathcal L}, \tilde {\mathtt{P}}_{\mathcal L}) \\
        &\quad= \sum_{\lambda_{\mathcal L} \in \Omega_{\mathcal L}} \abs{\prod_{j=1}^{L} \br{\Gamma_{0,j}(\lambda_{\mathcal L}) + \Gamma_{1,j}(\lambda_{\mathcal L})} - \prod_{j=1}^{L} \Gamma_{0,j}(\lambda_{\mathcal L})} \\
        &\quad= \sum_{\lambda_{\mathcal L} \in \Omega_{\mathcal L}} \abs{\sum_{b = 1}^{2^{L} - 1}\prod_{j=1}^{L} \Gamma_{b_{j}, j}(\lambda_{\mathcal L})} \\
        &\quad\leq \sum_{\lambda_{\mathcal L} \in \Omega_{\mathcal L}} \sum_{b = 1}^{2^{L} - 1} \prod_{j=1}^{L} \abs{ \Gamma_{b_{j}, j}(\lambda_{\mathcal L})}
    \end{align}
    Summing over $\Gamma_{0, j}$ yields $1$ due to normalization of $ \tilde {\mathtt{P}}_{\ell_{j}}(\lambda_{\ell_{j}})$ in Equation~\ref{eq:gamma_values}. However, summing over $\Gamma_{0, j}$ yields $(\abs{\Omega_{\ell_{j}}} - 1)/\abs{\Omega_{u_{j}}}$ exactly as in Theorem~\ref{thm:discrete_uniform_sampling}. Therefore,
    \begin{align}
        \Delta(\mathtt{P}_{\mathcal L}, \tilde {\mathtt{P}}_{\mathcal L}) 
        &\leq \sum_{k_1=1}^{L} \frac{\br{\abs{\Omega_{\ell_{k_{1}}}} - 1}}{\abs{\Omega_{u_{k_{1}}}}} + \frac{1}{2!}\sum_{k_1=1}^{L}\sum_{k_2=1}^{L} \frac{\br{\abs{\Omega_{\ell_{k_{1}}}} - 1}\br{\abs{\Omega_{\ell_{k_{2}}}} - 1}}{\abs{\Omega_{u_{k_{1}}}}\abs{\Omega_{u_{k_{2}}}}} + \cdots \label{eq:truncated_series}
    \end{align}
    In order to simplify Equation~\ref{eq:truncated_series}, let $C, K$ be defined as,
    \begin{align}
        C = \max\bc{\abs{\Omega_{\ell_{j}}} \mid 1 \leq j \leq L}, \quad K = \min\bc{\abs{\Omega_{u_{j}}} \mid 1 \leq j \leq L}.
        \label{eq:values_for_C_K}
    \end{align}
    Combining Equations~\ref{eq:error_is_only_latent},~\ref{eq:truncated_series}, and~\ref{eq:values_for_C_K}, one obtains the required result,
    \begin{align}
        \Delta(\mathtt{P}_{\mathcal V}, \tilde {\mathtt{P}}_{\mathcal V}) \leq \sum_{n = 1}^{L} \frac{1}{n!}\br{\frac{L(C - 1)}{K}}^{n}
    \end{align}
    To conclude the proof, one needs to prove the existence of a uniform functional model $(\mathcal G, \tilde {\mathcal F}_{\mathcal V}, \tilde {\mathcal P}_{\mathcal L})$ which reproduces $\tilde {\mathtt{P}}_{\mathcal V}$. To do so, substitute into Equation~\ref{eq:the_tilde_that_needs_to_match} the functional form of the rational approximations (Equation~\ref{eq:thm_discrete_uniform_sampling}) from Theorem~\ref{thm:discrete_uniform_sampling} for each $\ell_{j} \in \mathcal L$,
    \begin{align}
        \tilde {\mathtt{P}}_{\mathcal V}(x_{\mathcal V}) = \prod_{j \in 1}^{L} \sum_{\lambda_{\ell_j} \in \Omega_{\ell_{j}}} \bc{\frac{1}{\abs{\Omega_{u_{j}}}} \sum_{\omega_{u_{j}} \in \Omega_{u_{j}}} \delta(\lambda_{\ell_{j}}, g_{j}(\omega_{u_{j}}))}\delta(x_{\mathcal V}, \mathcal F_{\mathcal V}(\lambda_{\ell_{1}}\lambda_{\ell_{2}}\ldots\lambda_{\ell_{L}}) ).
    \end{align}
    Perform the sum over all latent valuations to remove the inner delta function,
    \begin{align}
        \tilde {\mathtt{P}}_{\mathcal V}(x_{\mathcal V}) = \prod_{j \in 1}^{L} \bs{\frac{1}{\abs{\Omega_{u_{j}}}} \sum_{\omega_{u_{j}} \in \Omega_{u_{j}}}}\delta(x_{\mathcal V}, \mathcal F_{\mathcal V}(g_{1}(\omega_{u_{1}})g_{2}(\omega_{u_{2}})\ldots g_{L}(\omega_{u_{L}})) ). \label{eq:the_uniform_functional_model}
    \end{align}
    Finally, one can recursively define the functions in $\tilde {\mathcal F}_{\mathcal V}$ to be such that $\tilde {\mathcal F}_{\mathcal V}(\omega_{\mathcal L'}) = {\mathcal F}_{\mathcal V}(g(\omega_{\mathcal L'}))$ and consequently Equation~\ref{eq:the_uniform_functional_model} defines the uniform functional model $(\mathcal G, \tilde {\mathcal F}_{\mathcal V}, \tilde {\mathcal P}_{\mathcal L})$ which reproduces $\tilde {\mathtt{P}}_{\mathcal V}$.
\end{proof}

\end{document}